\documentclass[sloppy,oribibl,envcountsame,envcountsect,11pt]{article}

\usepackage{amsfonts,amssymb,amsthm,amsmath}
\usepackage{hyperref}
\usepackage{mathtools}
\usepackage{comment}

\usepackage[margin=1in]{geometry}

\usepackage[utf8]{inputenc} 
\usepackage{amsmath,amssymb,amsthm}
\usepackage{algorithm}

\usepackage{float}
\usepackage{verbatim}

\usepackage{todonotes}

\DeclareMathOperator {\range}{Range}
\DeclareMathOperator {\col}{Col}

\title{On Efficient Noncommutative Polynomial Factorization via Higman Linearization \thanks{A preliminary version was presented at 
  the 37th Computational Complexity Conference, CCC'22, \cite{AJ22}.}}

\author{V. Arvind\thanks{Institute of Mathematical Sciences, Chennai,
    India and Chennai Mathematical Institute, Siruseri, Kelambakkam,
    India, \texttt{email: arvind@imsc.res.in}} \and Pushkar S
  Joglekar\thanks{Vishwakarma Institute of Technology, Pune, India,
    \texttt{email: joglekar.pushkar@gmail.com}}
    \thanks{Author would like to thank SERB for the funding through the MATRICS project, File no. MTR/2018/001214}
    }

\date{}

\newtheorem{theorem}{Theorem}[section]
\newtheorem{corollary}[theorem]{Corollary}
\newtheorem{definition}[theorem]{Definition}
\newtheorem{lemma}[theorem]{Lemma}

\newtheorem{claim}[theorem]{Claim}
\newtheorem{remark}[theorem]{Remark}

\newenvironment{proofof}[1]{\noindent{\it Proof of #1. }} {{\qed}}

\newtheorem*{theorem*}{Theorem}

\newtheorem{problem}[theorem]{Problem}

\usepackage{enumitem}

\def\qed{\hspace*{\fill} $\Box$\par\medskip}

\usepackage{graphicx}

\usepackage{amsmath,amssymb}

\newcommand{\op}[1]{\ensuremath{\operatorname{#1}}}

\newcommand{\F}{\mathbb{F}}
\renewcommand{\angle}[1]{\langle #1 \rangle}
\newcommand{\FX}{\F\angle{X}}

\newcommand{\poly}{\mathrm{poly}}

\newcommand{\rank}{\op{rank}}

\newcommand{\Q}{\mathbb{Q}}

\newcommand{\fR}{\FX}

\newcommand{\skewf}{\mathbb{F}\newbrak{X}}
\newcommand{\fF}{\skewf}

\newcommand{\prob}[1]{\textsc{#1}}

\DeclareFontEncoding{LS1}{}{}
\DeclareFontSubstitution{LS1}{stix}{m}{n}
\DeclareSymbolFont{symbols2stix}{LS1}{stixfrak} {m} {n}
\DeclareMathSymbol{\lparenless}{\mathopen} {symbols2stix}{"32}
\DeclareMathSymbol{\rparengtr}{\mathclose}{symbols2stix}{"33}

\newcommand{\newbrak}[1]{{\lparenless} #1 {\rparengtr}}

\usepackage{framed}

\begin{document}

\maketitle

\begin{abstract}

In this paper we study the problem of efficiently factorizing
polynomials in the free noncommutative ring
$\F\angle{x_1,x_2,\ldots,x_n}$ of polynomials in noncommuting
variables $x_1,x_2,\ldots,x_n$ over the field $\F$. We obtain
the following result:

\begin{itemize} 
\item[] Given a noncommutative algebraic branching
  program\footnote{This strengthens the main result in 
    earlier versions of this paper where the algorithm was only
    for noncommutative arithmetic formulas.} of size $s$ computing a
  noncommutative polynomial $f\in\F\angle{x_1,x_2,\ldots,x_n}$ as
  input, where $\F=\F_q$ is a finite field, we give a randomized
  algorithm that runs in time polynomial in $s, n$ and $\log_2q$ that
  computes a factorization of $f$ as a product $f=f_1f_2\cdots f_r$,
  where each $f_i$ is an irreducible polynomial that is output as a
  noncommutative algebraic branching program.
   
\item[] The algorithm works by first transforming the given algebraic
  branching program computing $f$ into a linear matrix $L$ using
  Higman's linearization of polynomials. We then factorize the linear
  matrix $L$ and recover the factorization of $f$. We use basic
  elements from Cohn's theory of free ideals rings combined with
  Ronyai's randomized polynomial-time algorithm for computing
  invariant subspaces of a collection of matrices over finite fields.
\end{itemize}

\end{abstract}

\newpage

\setcounter{page}{0}
\tableofcontents
\thispagestyle{empty}
\newpage

\section{Introduction}\label{intro}

Let $\F$ be any field and $X=\{x_1,x_2,\ldots,x_n\}$ be a set of $n$
free noncommuting variables. Let $X^*$ denote the set of all free
words (which are monomials) over the alphabet $X$ with concatenation
of words as the monoid operation and the empty word $\epsilon$ as
identity element.

The \emph{free noncommutative ring} $\FX$ consists of all finite
$\F$-linear combinations of monomials in $X^*$, where the ring
addition $+$ is coefficient-wise addition and the ring multiplication
$*$ is the usual convolution product. More precisely, let $f,g\in\FX$
and let $f(m)\in\F$ denote the coefficient of monomial $m$ in
polynomial $f$. Then we can write $f=\sum_m f(m) m$ and $g=\sum_m g(m)
m$, and in the product polynomial $fg$ for each monomial $m$ we have

\[
fg(m)=\sum_{m_1m_2=m} f(m_1)g(m_2).
\]

The \emph{degree} of a monomial $m\in X^*$ is the length of the
monomial $m$, and the degree $\deg f$ of a polynomial $f\in\FX$ is the
degree of a largest degree monomial in $f$ with nonzero coefficient.
For polynomials $f,g\in\FX$ we clearly have $\deg (fg) = \deg f + \deg
g$.

A \emph{nontrivial factorization} of a polynomial $f\in\FX$ is an
expression of $f$ as a product $f=gh$ of polynomials $g,h\in\FX$ such
that $\deg g > 0$ and $\deg h > 0$. A polynomial $f\in\FX$ is
\emph{irreducible} if it has no nontrivial factorization and is
\emph{reducible} otherwise.  For instance, all degree $1$ polynomials
in $\FX$ are irreducible. Clearly, by repeated factorization every
polynomial in $\FX$ can be expressed as a product of
irreducibles. 

In this paper we study the algorithmic complexity of polynomial
factorization in the free ring $\FX$. The factorization algorithm is
by an application of Higman's linearization process followed by
factorization of a matrix with linear entries (under some technical
conditions) using Cohn's factorization theory.

It is interesting to note that Higman's linearization process
\cite{Hig} has been used to obtain a deterministic polynomial-time
algorithm for the RIT problem. That is, the problem of testing if a
noncommutative rational formula (which computes an element of the free
skew field $\skewf$) is zero on its domain of definition
\cite{GGOW20,IQS17, IQS18,HW15}.

%We sketch the main ideas in our factorization algorithm.

\subsection{Overview of the results}

The main result of the paper is the following.

\begin{theorem*}[Main Theorem]
Given a multivariate noncommutative polynomial $f\in\F_q\angle{X}$ for
a finite field\footnote{We present the detailed randomized algorithm
  over large finite fields. In the case of small finite fields we
  obtain a deterministic $\poly(s,q,|X|)$ time algorithm with minor
  modifications.} $\F_q$ by a
noncommutative algebraic branching program of size $s$ as input, a
factorization of $f$ as a product $f=f_1f_2\cdots f_r$ can be computed
in randomized time $\poly(s,\log_2q,|X|)$, where each
$f_i\in\F_q\angle{X}$ is an irreducible polynomial that is output as
an algebraic branching program.
\end{theorem*} 
%\pushkar{updated statement}

The proof has three broad steps described below.

\begin{itemize}
\item \textbf{Higman linearization and Cohn's factorization theory}~~
  Briefly, given a noncommutative polynomial $f\in\FX$, we can
  transform it into a linear matrix $L$ such that $f\oplus I = PLQ$,
  where $P$ is an upper triangular matrix with polynomial entries and
  all $1$'s diagonal and $Q$ is a lower triangular matrix with
  polynomial entries and all $1$'s diagonal, $P$ and $Q$ are the
  matrices implementing the sequence of row and column operations
  required for the Higman linearization process. Now, as observed
  by Garg et al \cite{GGOW20} (in their work on rational identity
  testing), if $f$ is given by a noncommutative formula as input
  then the above Higman linearization can be carried out in polynomial
  time. We strengthen their observation with a modified Higman
  linearization process that we call Block-Higman linearization (to
  emphasize that the modification works with matrix blocks) and
  show the more general result that if $f$ is given by an ABP as input
  we can still compute its Higman linearization as defined above
  in polynomial time.

%\item \item \textbf{Cohn's factorization theory}~~

\item \textbf{Ronyai's common invariant subspace algorithm}~~ Next,
  the most important tool algorithmically, is Ronyai's algorithm for
  computing common invariant subspaces of a collection of matrices
  over finite fields \cite{Ronyai2}. We show that Ronyai's common
  invariant subspace algorithm can be repeatedly applied to factorize
  a linear matrix $L=A_0+\sum_{i=1}^n A_i x_i$, into a product of
  irreducible linear matrices provided $A_0$ is invertible and
  $[A_1 A_2\cdots A_n]$ has full row rank or
  $[A_1^T A_2^T\cdots A_n^T]^T$ has full column rank. The later
  conditions are called as right and left monicity of the linear
  matrix $L$ respectively. With some technical work we can ensure
  these conditions for a linear matrix $L$ that is produced from a
  polynomial $f$ by Higman linearization. Then Ronyai's algorithm
  yields the factorization of $L$ into a product of irreducible linear
  matrices (upto multiplication by units). Here, Cohn's theory of
  factorization of noncommutative linear matrices gives us sufficient
  useful information about the structure of irreducible linear matrices.

\item \textbf{Recovering the factors of $f$}~~ Finally, we design a
  simple linear algebraic algorithm for trivializing a matrix product
  $AB=0$, where $A$ is a linear matrix and $B$ is a column vector of
  polynomials from $\FX$, using which we are able to extract the
  irreducible factors of $f$ from the factors of $L$. An invertible
  matrix $M$ with polynomial entries \emph{trivializes} the relation
  $AB=0$ if the modified relation $(AM)(M^{-1}B)=0$ has the property
  that for every index $i$ either the $i^{th}$ column of $AM$ is zero
  or the $i^{th}$ row of $M^{-1}B$ is zero. While such matrices $M$
  exist for any matrix product $AB=0$ with entries from $\FX$, we
  obtain an efficient algorithm in the special case when $A$ is linear
  and $B$'s entries are polynomials computed by small arithmetic
  circuits. This special case is sufficient for our application.
\end{itemize}

There are some additional technical aspects we need to deal with. 
%%applying Cohn's factorization theory to efficiently transform the
%%problem of factorizing $L$ to the problem of computing a common
%%invariant subspace for a collection of scalar matrices over
%%$\F_q$.
Let $L=A_0+\sum_{i=1}^n A_i x_i$ be the linear matrix obtained from
$f\in\F_q\angle{X}$ by Higman linearization, where
$X=\{x_1,x_2,\ldots,x_n\}$ and $A_i\in \F_q^{d\times d}, 0\le i\le
n$. If $A_0$ is an invertible matrix then it turns out that the
problem of factorizing $L$ can be directly reduced to the problem of
finding a common invariant subspace for the matrices $A_0^{-1}A_i,
1\le i\le n$.  In general, however, $A_0$ is not invertible. Two cases
arise:
\begin{itemize}
\item[(a)] The polynomial $f$ is \emph{commutatively nonzero}. That
  is, it is nonzero on $\F_q^n$ (or on $\F^n$ for a small extension
  field $\F$). In this case, by the DeMillo-Lipton-Schwartz-Zippel
  Lemma~\cite{DL78,Sch80, Zip79}, we can do a linear shift of the
  variables $x_i\leftarrow x_i+\alpha_i$ in the polynomial $f$, for
  $\alpha_i$ randomly picked from $\F_q$ (or $\F$). Let the resulting
  polynomial be $f'$ and let its Higman linearization be $L_{f'}$. In
  $L_{f'}$ the constant matrix term $A'_0$ will be invertible with
  high probabilty, and the reduction steps outlined above will work
  for $L_{f'}$. Furthermore, from the factorization of $f'$ we can
  efficiently recover the factorization of $f$. Section~\ref{cnz-sec}
  deals with Case~(a), with Theorem~\ref{cnzthm} summarizing the
  algorithm for factorizing $f$. Theorem~\ref{lfact1} describes the
  algorithm for factorization of the linear matrix $L_{f'}$, and the
  factor extraction lemma (Lemma~\ref{extract}) allows us to
  efficiently recover the factorization of $f'$ from the factorization
  of $L_{f'}$.
  
\item[(b)] In the second case, suppose $f$ is zero on all scalars. Then, for example by
  Amitsur's theorem \cite{Ami66}, for a random matrix substitution
  $x_i\leftarrow M_i\in \F^{2s\times 2s}$ the matrix
  $f(M_1,M_2,\ldots,M_n)$ is \emph{invertible} with high probability,
  where $s$ is the formula size of $f$.\footnote{Amitsur's theorem
    strengthens the Amitsur-Levitski theorem \cite{AL50} often used in
    noncommutative PIT algorithms \cite{BW05}.} \footnote{In the actual algorithm we pick the matrices $M_i$ using a result from  \cite{DM17}}   Accordingly, we can
  consider the factorization problem for shifted and dilated linear
  matrix $L'=A_0\otimes I_{\ell} + \sum_{i=1}^n A_i\otimes (Y_i+M_i)$ which             
  will have the constant matrix term invertible, where each $Y_i$ is  
  an $\ell \times \ell$ matrix of distinct noncommuting variables, where $\ell=2s$.  Recovering
  the factorization of $L$ from the factorization of $L'$ requires
  some additional algorithmic work based on linear algebra. A lemma
  from \cite{HKV20} (refer Section \ref{cz-sec} and the Appendix for the details) turns
  out to be crucial here. The algorithm handling Case~(b) is described
  in Section~\ref{cz-sec}. Indeed, the new aspect of the algorithm is
  factorization of the dilated matrix $L'$ from which we recover the
  factorization of the Higman linearization $L_f$ of $f$. The
  remaining algorithm steps are exactly as in Section~\ref{cnz-sec}.
\end{itemize}  

\subsection{Small Finite fields}\label{small-field}

We now briefly explain the deterministic $\poly(s,q,|X|)$ time
factorization algorithm (when $\F_q$ is small). There are two places
in the factorization algorithm outlined above where randomization is
used: first, to obtain a matrix tuple $(M_1,M_2,\ldots,M_n)$ such that
$f(M_1,M_2,\ldots,M_n)$ is invertible, which ensures that the constant
matrix term of the linear matrix $L'$ is invertible. When
$q=\Omega(d)$, where $d= \deg f$, it suffices to randomly pick
$M_i\in\F_q^{2s\times 2s}$. However, if $q<d$ we can choose entries of
the matrices $M_i$ from a small extension field $\F_{q^k}$ such that
$q^k=\Omega(d)$.  Thereby, we will obtain factorization of $L'$ and
subsequently that of the polynomial $f$ over the extension field
$\F_{q^k}$. However, we can use the fact that the finite field
$\F_{q^k}$ can be embedded using the regular representation of the
elements of $\F_{q^k}$ in the matrix algebra $\F_q^{k\times k}$. Thus,
we can obtain from $(M_1,M_2,\ldots,M_n)$ a matrix tuple
$(M'_1,M'_2,\ldots,M'_n)$ with $M'_i\in\F_q^{2sk\times 2sk}$ such that
$f(M'_1,M'_2,\ldots,M'_n)$ is invertible. This will ensure that the
linear matrix $L'$ can be factorized over the field $\F_q$ which will
allow us to obtain a complete factorization of $f$ into irreducible
factors over $\F_q$.

In order to get a deterministic polynomial-time algorithm for finding
such matrices $M'_i, 1\le i\le n$ we will use the fact that the
polynomial $f$ is given by a small noncommutative formula and hence
has a small algebraic branching program. Then, using ideas from
\cite{RS05,F14,ACDM20} we can easily find such matrices $M'_i$ in
deterministic polynomial time.

Next, we notice that Ronyai's algorithm for finding common invariant
subspaces of matrices over $\F_q$ is essentially a polynomial-time
reduction to univariate polynomial factorization over $\F_q$. We can
use Berlekamp's deterministic $\poly(q,D)$ algorithm for the
factorization of univariate degree $D$ polynomials over
$\F_q$. Putting it together, we can obtain a deterministic
$\poly(s,q,|X|)$ time algorithm for factorization of
$f\in\F_q\angle{X}$ as a product of irreducible factors over
$\F_q$.\smallskip

\begin{comment}
Infact, it turns out that, building on ideas from Arvind
et.al. \cite{ACDM20} combined with determinsitic algorithm for PIT of
Read Once ABPs by Forbes and Shpilka \cite{} and using regular
representation of $\F_{q^k}$ as $k \times k$ matrices over
$\mathbb{F}_q$ we can efficiently find a matrix substitution
$\overline{N}= (N_1, N_2, \ldots, N_n)$ such that entries of $N_i$'s
are from base field $\F_q$ and the size of $N_i$ is at most $k$ times
as that of the original substitution matrices $M_i$'s and the linear
matrix obtained by using the substitution $\overline{N}$ instead of
$\overline{M}$ also has invertible constant term. As a result we can
obtain factorization of $f$ over $\F_q$. Moreover, for the case of
small finite fields, Ronyai's algorithm works in deterministic
polynomial time. Putting it together, in the case of the small finite
fields, we get deterministic polynomial time factorization algorithm
for factorization of the polynomial $f$ given by noncommutative
formula. 
\end{comment}

%\pushkar{added the remark}

\subsection{Finite fields versus Rationals}\label{finite-vs-rationals}

Unfortunately, the algorithm outlined above does not yield an
efficient algorithm for noncommutative polynomial factorization over
rationals. The bottlneck is the problem of computing common invariant
subspaces for a collection of matrices over $\Q$. Ronyai's algorithm
for the problem over finite fields \cite{Ronyai2} builds on the
decomposition of finite-dimensional associative algebras over
fields. Given an algebra $\mathcal{A}$ over a finite field $\F_q$ the
algorithm decomposes $\mathcal{A}$ as a direct sum of minimal left
ideals of $\mathcal{A}$ which is used to find nontrivial common
invariant subspaces. However, as shown by Friedl and Ronyai
\cite{FR85}, over rationals the problem of decomposing a \emph{simple}
algebra as a direct sum of minimal left ideals is at least as hard as
factoring square-free integers.

\subsection{Related research}

The study of factorization in noncommutative rings is systematically
investigated as part of Cohn's general theory of noncommutative free
ideal rings \cite{Cohnfir,Cohnintro} which is based on the notion of
the weak algorithm. In fact, there is a hierachy of weak algorithms
generalizing the division algorithm for commutative integral domains
\cite{Cohnfir}.\smallskip

\noindent\textit{Algorithmic:}~~To the best of our knowledge, the
complexity of noncommutative polynomial factorization has not been
studied much, unlike the problem of commutative polynomial
factorization \cite{vzg-book,K89,KT90}.  Prior work on the complexity
of noncommutative polynomial factorization we are aware of is
\cite{AJR18} where efficient algorithms are described for the problem
of factoring \emph{homogeneous} noncommutative polynomials (which
enjoy the unique factorization property, and indeed the algorithms in \cite{AJR18} crucially use the unique factorization property). When the input homogeneous
noncommutative polynomial has a small noncommutative arithmetic
circuit (even given by a black-box as in Kaltofen's algorithms
\cite{K89, KT90}) it turns out that the problem is efficiently
reducible to commutative factorization by set-multilinearizing the
given noncommutative polynomial with new commuting variables. This
also works in the black-box setting and yields a randomized
polynomial-time algorithm which will produce as output black-boxes for
the irreducible factors (which will all be homogeneous). When the
input homogeneous polynomial is given by an algebraic branching
program there is even a deterministic polynomial-time factorization
algorithm. Indeed, the noncommutative factorization problem in for
homogeneous polynomials efficiently reduces to the noncommutative PIT
problem \cite{AJR18}, analogous to the commutative case \cite{KSS14},
modulo the randomness required for univariate polynomial factorization
in the case of finite fields of large characteristic. The motivation
of the present paper is to extend the above results to the
inhomogeneous case.\smallskip

\noindent\textit{Mathematical:}~~From a mathematical perspective,
building on Cohn's work there is a lot of research on the study of
noncommutative factorization. For example, \cite{BS15,BHL17} focus on
the lack of unique factorization in noncommutative rings and study the
structure of multiple factorizations. The research most relevant to
our work is the study of noncommutative analogues of the
Nullstellensatz by Helton, Klep and Volcic \cite{HKV18,HKV20}. In  
these papers the authors study the free singularity locus of a
noncommutative polynomial $f\in\FX$ where $\F$ is an algebraically
closed field of characteristic zero (in \cite{HKV20} mostly they
consider complex numbers). This is the set of all matrix tuples
$\bar{M}\in\mathcal{L}_n(f)$ (in all matrix dimensions $d$) where
$\mathcal{L}_n(f)=\{\bar{M}\mid \det f(\bar{M}) = 0,$ where $\bar{M}$
is an $n$-tuple of matrices$\}$.  It turns out that $f\in\FX$ is
irreducible if and only if for all $d\ge d_0$ for some $d_0$ the
hypersurface $\mathcal{L}_d(f)$ is irreducible which in turn holds iff
$\det f(\bar{X})$ is an irreducible commutative polynomial, where
$\bar{X}$ are generic matrices with commuting variables of dimension
$d\ge d_0$. However, $d_0$ turns out to be exponentially large.

\smallskip

\noindent\textbf{Plan of the paper.}~ In Section~\ref{prelim} we
present basic definitions and the background results from Cohn's work on factorization. In Section~\ref{bas-res-sec} we further present some results from Cohn's work relevant to the paper. In
Section~\ref{cnz-sec} we present the factorization algorithm for
polynomials $f$ that does not vanish on scalars and in
Section~\ref{cz-sec} we present the algorithm for the general case.

\section{Preliminaries}\label{prelim}
%background

In this section we give some basic definitions and results relevant to
the paper, mainly from Cohn's theory of factorization.  Analogous to
integral domains and unique factorization domains in commutative ring
theory, P.M.~Cohn \cite{Cohnfir,Cohnintro} has developed a theory for
noncommutative rings based on the weak algorithm (a noncommutative
generalization of the Euclidean division algorithm) and the notion of
free ideal rings. We present the relevant basic definitions and
results, specialized to the ring $\FX$ of noncommutative polynomials
with coefficients in a (commutative) field $\F$, and also for matrix
rings with entries from $\FX$.

The results about $\FX$ in Cohn's text \cite[Chapter 5]{Cohnfir} are
stated uniformly for algebraically closed fields $\F$. However, those
we discuss hold for any field $\F$ (in particular for $\F_q$ or a
small degree extension of it). The proofs are essentially based on
linear algebra.

%% In this paper, for notational convenience, we will denote by $\fR$ the
%% free polynomial ring $\mathbb{F}\langle X \rangle$, and use $\fF$ to
%% denote the corresponding free skew field $\skewf$. Let $\fR^{m\times
%%   m}$ denote the ring of $m\times m$ matrices over $\fR$.

%\begin{remark}
%Write here about Cohn's factorization result for polynomials in $\FX$,
%fact that it is not unique, concept of stable associates, etc, citing
%reference to relevant theorems in Cohn....
%\end{remark}   

Since we will be using Higman's linearization \cite{Hig} to factorize
noncommutative polynomials, we are naturally lead to studying the
factorization of linear matrices in $\fR^{d\times d}$ using Cohn's
theory.

\begin{definition}{\rm\cite{Cohnfir}}
A matrix M in $\fR^{d\times d}$ is called \emph{full} if it has
(noncommutative) rank $d$. That is, it cannot be decomposed as a
matrix product $M= M_1\cdot M_2$, for matrices $M_1 \in \fF^{d \times
  e}$ and $M_2 \in \fF^{e \times d}$ with $e <d$.
\end{definition}

\begin{remark}
 Based on the notion of noncommutative matrix rank \cite{Cohnfir}, the
 square matrix $M\in\fR^{d\times d}$ is full precisely when it is
 invertible in the skew field $\fF$. That is, $M$ is full if and only
 if there is a matrix $N\in\fF^{d\times d}$ such that $MN=NM=I_d$, where $I_d$ is $d \times
d$ identity matrix.
\end{remark}

We note the distinction between full matrices and units in the matrix
ring $\fR^{d\times d}$.

\begin{definition}
A matrix $U \in \fR^{d\times d}$ is a \emph{unit} if there is a matrix
$V\in \fR^{d\times d}$ such that $UV=VU=I_d$, where $I_d$ is $d \times
d$ identity matrix.
\end{definition}

Clearly, units in $\fR^{d \times d}$ are full. Examples of units in
$\fR^{d\times d}$, which have an important role in our factorization
algorithm, are upper (or lower) triangular matrices in $\fR^{d \times
  d}$ whose diagonal entries are all \emph{nonzero scalars}. Full
matrices, in general, need not be units: for example, the $1\times 1$
matrix $x$, where $x$ is a variable, is full but it is not a unit in
the ring $\fR^{1\times 1}=\fR$.

\begin{remark}
Full non-unit matrices are essentially non-unit non-zero-divisors.
For the factorization of elements in $\fR^{d \times d}$, units are
similar to scalars in the factorization of polynomials in polynomial
rings. Cohn's theory \cite{Cohnfir} considers factorizations of full
non-unit elements in $\fR^{d \times d}$.
\end{remark}
%%\footnote{But units that are upper triangular or lower triangular
%%  matrices play a role in our factorization algorithm.}
We next define \emph{atoms} in $\fR^{d\times d}$, which are
essentially the irreducible elements in it.

\begin{definition}
A full non-unit element $A$ in $\fR^{d \times d}$ is an \emph{atom} if
$A$ cannot be factorized as $A=A_1A_2$ for full non-unit matrices
$A_1, A_2$ in $\fR^{d \times d}$.
\end{definition}

Noncommutative polynomials do not have unique factorization in the usual
sense of commutative polynomial factorization.\footnote{However, as
  shown by Cohn, using the notion of stable associates there is a more
  general sense in which noncommutative polynomials have ``unique''
  factorization \cite{Cohnfir}.}  A classic example \cite{Cohnfir} is
the polynomial $x+xyx$ with its two different factorizations
\[
x+xyx = x(1+yx) = (1+xy)x,
\]
where $1+xy$ and $1+yx$ are distinct irreducible polynomials.

\begin{definition}
Elements $A \in \fR^{d \times d}$ and $B \in \fR^{d' \times d'}$ are
called \emph{stable associates} if there are positive integers $t$ and
$t'$ such that $d+t=d'+t'$ and units $P, Q\in\fR^{(d+t)\times (d+t)}$
such that $A\oplus I_t=P(B \oplus I_{t'})Q$.
\end{definition}

It is easy to check that the polynomials $1+xy$ and $1+yx$ are stable
associates.

Notice that if $A$ and $B$ are full non-unit matrices that are stable
associates then $A$ is atom if and only if $B$ is atom.  Furthermore,
we note that stable associativity defines an equivalence relation
between full matrices over the ring $\fR$.

We observe that the problem of checking if two polynomials in $\fR$ given as arithmetic formulas 
are stable associates or not has an efficient randomized algorithm
(Lemma~\ref{stable-algo}). 

%Furthermore, if $A$ is not atom then given a nontrivial factorization
%of $A$ we can obtain

%\begin{remark}
%We can state here fact that if $A$ and $B$ are stable associates then
%A is atom iff B is atom, We can get non-trivial matrix factorization
%of A from that of B and conversely. stable associativity is
%equivalence relation. Role it plays in factorization of matrix
%polynomials....
%\end{remark}

%\begin{remark} Write here about factorization of matrix polynomials,
%fact that it is not unique, citing refer%ence to relevant theorems in
%Cohn....  \end{remark}

Now we turn to the problem of noncommutative polynomial factorization.
By Higman's linearization \cite{Hig,Cohnfir}, given a polynomial
$f \in \fR$ there is a positive integer $\ell$ such that $f$ is stably
associated with a \emph{linear matrix} $L \in \fR^{\ell \times \ell}$,
that is to say, the entries of $L$ are affine linear
forms.\footnote{More generally, by Higman's linearization any matrix
  of polynomials $M$ is stably associated with a linear matrix
  $L \in \fR^{\ell \times \ell}$ for some $\ell$.}  Higman's
linearization process is a simple algorithm obtaining the linear
matrix $L$ for a given $f$, and it plays a crucial role in our
factorization algorithm. We first describe it and then recall an
effective version \cite{GGOW20} which gives a simple polynomial-time
algorithm to compute $L$ when $f$ is given as a non-commutative
arithmetic formula. Then we state our stronger result showing that
even if $f$ is given by an algebraic branching program as input
we can compute its Higman linearization in deterministic polynomial
time.

\subsection*{Higman's linearization process}

%We explain the basic ideas involved along with an overview of our
%algorithmic result. After that we will present the details.

We describe a single step of the linearization process. Given an
$m \times m$ matrix $M$ over $\fR$ such that $M[m,m]=f+g\times h$,
apply the following:
\begin{itemize}
\item Expand $M$ to an  $(m+1)\times (m+1)$ matrix by adding a new last
  row and last column with diagonal entry $1$ and remaining new entries
  zero:
\[
\left[
\begin{array}{c|c}
M & 0 \\
\hline
0 & 1
\end{array}
\right]. 
\]
\item Then the bottom right $2\times 2$ submatrix is transformed as
  follows by elementary row and column operations

\[
\left(
\begin{array}{cc}
f+gh & 0 \\
0 & 1
\end{array}
\right)\rightarrow
\left(
\begin{array}{cc}
f+gh & g \\
0 & 1
\end{array}
\right)\rightarrow
\left(
\begin{array}{cc}
f & g \\
-h & 1
\end{array}
\right)
\]

\end{itemize}

Given a polynomial $f\in\FX$ by repeated application of the above step
we will finally obtain a \emph{linear matrix} $L=A_0+\sum_{i=1}^n A_i
x_i$, where each $A_i, 0\le i\le n$ is an $\ell\times \ell$ over $\F$,
for some $\ell$. The following theorem summarizes its properties.

\begin{theorem}[Higman Linearization]{\rm\cite{Cohnfir}}\label{higthm}
  Given a polynomial $f\in\fR$, there are matrices $P, Q\in\fR^{\ell \times \ell}$
  and a linear matrix $L\in\fR^{\ell\times \ell}$ such that
\begin{equation}\label{higeq}
%PAQ ~=~  L \text{ where } A ~=~ 
\left(
\begin{array}{c|c}
f & 0 \\
\hline
0 & I_{\ell-1}
\end{array}
\right) ~=~PLQ
\end{equation}
with $P$ upper triangular, $Q$ lower triangular, and the diagonal
entries of both $P$ and $Q$ are all $1$'s (hence, $P$ and $Q$ are both
units in $\fR^{\ell \times \ell}$).
\end{theorem}

Instead of a single $f$, we can apply Higman linearization to a matrix
of polynomials $M\in\fR^{m\times m}$ to obtain a linear matrix $L$
that is stably associated to $M$.  We first recall the algorithmic
version of Garg et al.~\cite{GGOW20} in this general form.

\begin{theorem}\label{ehigman}{\rm\cite[Proposition A.2]{GGOW20}}
Let $M \in \fR^{m \times m}$ such that $M_{i,j}$ is computed by a
non-commutative arithmetic formula of size at most $s$ and bit
complexity at most $b$. Then, for $k=O(s)$, in time $\poly(s,b)$ we
can compute the matrices $P, Q$ and $L$ in $\fR^{\ell \times \ell }$
of Higman's linearization such that
\[
\left(
\begin{array}{c|c}
M & 0 \\
\hline
0 & I_{k}
\end{array}
\right) ~=~PLQ.
\], where $\ell= m+k$.
Moreover, the entries of the matrices $P$ and $Q$ as well as $P^{-1}$
and $Q^{-1}$ are given by polynomial-size algebraic branching programs
which can also be obtained in polynomial time.
\end{theorem}

We will sometimes denote the block diagonal matrix
$\left(
\begin{array}{c|c}
M & 0 \\
\hline
0 & I_{k}
\end{array}
\right)$ by $M\oplus I_k$.

We now state our strengthening of Theorem~\ref{ehigman} which enables
us to factorize noncommutative polynomials given as algebraic
branching programs. The complete proof is presented in the appendix.

\begin{theorem}\label{block-higman}
  Let $M \in \fR^{m \times m}$ such that each entry $M_{i,j}$ is a
  polynomial computed by a non-commutative algebraic branching program
  of size at most $s$ and bit complexity at most $b$. Then, for
  $k=O(s)$, in time $\poly(s,b)$ we can compute the matrices $P, Q$
  and $L$ in $\fR^{\ell \times \ell }$ of Higman's linearization such
  that
\[
\left(
\begin{array}{c|c}
M & 0 \\
\hline
0 & I_{k}
\end{array}
\right) ~=~PLQ.
\], where $\ell= m+k$.
Moreover, the entries of the matrices $P$ and $Q$ as well as $P^{-1}$
and $Q^{-1}$ are given by polynomial-size algebraic branching programs
which can also be obtained in polynomial time.
\end{theorem}

%Moreover, given access to the formulas computing individual entries of
%$M$, we can efficiently construct ABPs for $P$, $Q$ (in deterministic
%$poly(n,s,b)$ time). Also, since $P$ and $Q$ have all diagonal entries
%equal to $1$, it implies we can efficiently compute $P^{-1}$ and
%$Q^{-1}$ too.

As $P$ and $Q$ are units with diagonal entries all $1$'s, the matrix
$M$ is full iff the linear matrix $L$ is full. Also, the scalar matrix
$M(\overline{0})$ (obtained by setting all variables to zero) is
invertible iff the scalar matrix $L(\overline{0})$, similarly
obtained, is invertible.\\

\subsection*{Invariant Subspaces and Ronyai's Algorithm}

\begin{definition}
Let $A_1, \ldots, A_n \in \mathbb{F}^{d \times d}$. A subspace
$V\subseteq \F^n$ is called as common invariant subspace of $A_1,
\ldots, A_n$ if $A_i v \in V$ for all $i\in [n]$ and $v\in V$. 
\end{definition}

Clearly $0$ and $\F^n$ are, trivially, common invariant subspaces for
any collection of matrices. The algorithmic problem is to find a
\emph{non-trivial} common invariant subspace if one exists. Ronyai
\cite{Ronyai2} gives a randomized polynomial-time algorithm for this problem
when $\F$ is finite field.

\begin{theorem}\label{thm-ronyai}{\rm\cite{Ronyai2}}
Given $A_1, \ldots, A_n \in \mathbb{F}_q^{d \times d}$ there is a
randomized algorithm running in time polynomial in $n,d, \log q$ that
computes with high probability a non-trivial common invariant subspace
of $A_1, \ldots, A_n$ if such a subspace exists, and outputs ``no''
otherwise.
\end{theorem}
  
\begin{remark}
 We should note here, the classical Burnside's theorem \cite{Bur05} for matrix algebras over algebraically closed fields. It essentially shows that the algebra generated by $A_1, A_2, \ldots, A_n$ is the full matrix algebra iff there is no nontrivial common invariant subspace. \end{remark}  
  
\begin{remark}
    As already mentioned in the introduction, Friedl and Ronyai
  \cite{FR85} have shown that over rationals the problem is at least as
  hard as factoring square-free integers, and hence likely to be
  intractable.
  
\end{remark}

\subsection*{Noncommutative Formulas, Algebraic branching programs} 

Next we recall standard definitions of a noncommutative formulas and
noncommutative algebraic branching programs (ABPs). More details about
noncommutative arithmetic computation can be found in Nisan's work
\cite{N91}:
  
A \emph{noncommutative arithmetic circuit} $C$ over a field $\F$ and
indeterminates $x_1,x_2,\ldots,x_n$ is a directed acyclic graph (DAG)
with each node of indegree zero labeled by a variable or a scalar
constant from $\F$: the indegree $0$ nodes are the input nodes of the
circuit. Internal nodes, representing gates of the circuit, are of
indegree two and are labeled by either a $+$ or a $\times$ (indicating
the gate type). Furthermore, the two inputs to each $\times$ gate are
designated as left and right inputs prescribing the order of gate gate
multiplication. Each internal gate computes a polynomial (by adding or
multiplying its input polynomials), where the polynomial computed at
an input node is just its label.  A special gate of $C$ is the
\emph{output} and the polynomial computed by the circuit $C$ is the
polynomial computed at its output gate. An arithmetic circuit is a
\emph{formula} if the fan-out of every gate is at most one.\smallskip

A noncommutative \emph{algebraic branching program} (ABP) is a layered
directed acyclic graph with one source and one sink. The vertices of
the graph are partitioned into layers numbered from $0$ to $d$, where
edges may only go from layer $i$ to layer $i+1$. The source is the
only vertex at layer $0$ and the sink is the only vertex at layer
$d$. Each edge is labeled with a linear linear form in the
noncommuting variables $x_1, x_2, \ldots, x_n$ The size of the ABP is
the number of vertices. The polynomial in $\F\angle{X}$ computed by
the ABP is defined as follows: the sum over all source-to-sink paths
of the product of the linear forms by which the edges of the path are
labeled.

%%Clearly, an ABP with $d$ layers computes a polynomial in $\F\angle{X}$
%%of degree at most $d$.

\begin{comment}
If instead of homogeneous linear forms we use affine linear forms as
edge labels in the branching program, it will compute a possibly
inhomogeneous polynomial. For more details on noncommutative
arithmetic circuits, formulas, ABPs, the reader is referred to
\cite{Nis91}.
\end{comment}

\section{Some Basic Results}\label{bas-res-sec}

In this section we present some basic results required for our
factorization algorithm.

%\textbf{HERE WE WILL SPLIT THE PRELIM SECTION}\\

\subsection*{Monic linear matrices}

\begin{definition}\label{def-monic}{\rm\cite{Cohnfir}}
Let $L=A_0+A_1x_1+ \ldots + A_nx_n\in\fR^{d\times d}$ be a linear
matrix, where each $A_i$ is a $d\times d$ scalar matrix over
$\F$. Then $L$ is called \emph{right monic} if the $d\times nd$ scalar
matrix $[A_1~A_2~\ldots~A_n]$ has full row rank. Equivalently, if
there are matrices $B_1, \ldots, B_n \in \mathbb{F}^{d \times d}$ such
that $\Sigma_{i=1}^n A_iB_i = I_d$ (i.e. the matrix
$[A_1~A_2~\ldots~A_n]$ has right inverse).

Similarly, $L$ is \emph{left monic} if the $nd\times d$ matrix $[A_1^T~A_2^T~\ldots~A_n^T]^T$ 
has full column rank. $L$ is called \emph{monic} if it is
both left and right monic.

%$\left(
%\begin{array}{c}
%A_1\\
%A_2\\
%\vdots\\
%A_n
%\end{array}
%\right)$ 
\end{definition}

The next two results from Cohn \cite{Cohnfir} are important properties
of monic linear matrices.

\begin{lemma}\label{monic-nonunit}{\rm\cite{Cohnfir}}
A right (or left) monic linear matrix in $\fR^{d\times d}$ is not a unit in
$\fR^{d\times d}$.
\end{lemma}

\begin{proof}
Let $L=A_0 + \sum_{i=1}^n A_i x_i$ be right monic, where each $A_i\in
\F^{d\times d}$. By definition, there are matrices $B_i\in\F^{d\times
  d}, 1\le i\le d$ such that $\sum_{i=1}^d A_iB_i=I_d$. Now, suppose
$L$ is a unit. Then there is a matrix $C\in\fR^{d\times d}$ such that
$CL=I_d$. Let the maximum degree of polynomials occurring in $C$ be
$k$, and let $\hat{C}\in\fR^{d\times d}$ denote the degree $k$
component of $C$ (so each nonzero entry of $\hat{C}$ is a homogeneous
polynomial of degree $k$). Clearly, $\hat{C}\cdot (\sum_{i=1}^n A_i
x_i)=0$. The homogeneity of $\hat{C}$'s entries implies that $\hat{C}
A_i=0$ for each $i$. Hence, $\sum_{i=1}^n \hat{C}A_i B_i =0$ which
implies $\hat{C}=0$, contradicting the assumption that
$C\in\fR^{d\times d}$ is the inverse of $L$. The case when $L$ is left
monic is symmetric.
\end{proof}

Let $f\in\fR$ be a nonzero polynomial and $L$ be a linear matrix
obtained from $f$ by Higman linearization as in Equation~\ref{higeq}.
Clearly, $L$ is a full linear matrix. We show that we can transform
$L$ to obtain a full and right (or left) monic linear matrix $L'$ that is stably
associated to $f$. Furthermore, we can efficiently compute $L'$ and
the related transformation matrices.

%% This is adapted from a more general result \cite[Theorem
%%   5.8.3]{Cohnfir}.

%We state in a form that is directly relevant to
%polynomial factorization.

%It turns out that if the dimension $\ell$ is minimized then the linear
%matrix $L$ is monic.

%[Cohn's theorem 5.8.3 rephrased \cite{Cohnfir}]

\begin{theorem}{\rm\cite{Cohnfir}}\label{full-monic}
  Let $L = A_0 + \sum_{i=1}^n A_i x_i$ be a full linear matrix in
  $\fR^{d\times d}$ obtained by Higman linearization from a non
  constant polynomial $f\in\fR$. Then there are deterministic
  $\poly(n,d,\log_2 q)$ time algorithms that compute units $U, U'\in
  \fR^{d\times d}$ and invertible scalar matrices $S,
  S'\in\F_q^{d\times d}$ such that:
\begin{enumerate}  
\item $ULS= L'\oplus I_r$, and $L'$ is right monic. Moreover, if $L$
  is not right monic then $r>0$.
\item $S'LU'=L'\oplus I_{r'}$, and $L'$ is left monic. Moreover, if $L$ is
  not left monic then $r'>0$.
\end{enumerate}  
\end{theorem}

\begin{proof}
We prove only the first part. The second part has an essentially
identical proof.

We present a proof with a polynomial-time algorithm for computing
$L'$. If $L$ is already right monic there is nothing to
show. Otherwise, the row rank of the matrix $B = [A_1~A_2~\cdots~A_n]$
is strictly less than $d$. By row operations we can drive at least one
row of $B$ to zero. So, there is an invertible scalar matrix $U_1\in
\F^{d\times d}$ such that $U_1B$ has its last row as zeros. Now
$U_1A_0$ must have its last row non-zero since $L$ is a full linear
matrix. So the last row of $U_1L$ has only scalar entries and at least
one of these is non-zero. By a column swap applied to $U_1L$ we can
bring this non-zero scalar $\alpha$ in the $(d,d)^{th}$
position. Hence, the $(d,d)^{th}$ entry of $U_1LS_1$ is nonzero, where
$S_1$ is the matrix implementing the column swap. Now, with suitable
row operations using the last row, we can make all entries above the
$(d,d)^{th}$ entry of the $d^{th}$ column zero. Applying column
operations we can make all entries of the $d^{th}$ row to the left of
the $(d,d)^{th}$ entry zero. The resulting matrix is of the form
$RU_1LS_1S' = \tilde{L}\oplus 1$, where the unit $R$ is a linear
matrix and $S'$ is an invertible scalar matrix implementing the row
and column operations.

If $\tilde{L}$ is not right monic, we can recursively apply the above
procedure on $\tilde{L}$ until we finally obtain a unit
$\tilde{U}\in\fR^{d\times d}$ and a scalar invertible matrix
$\tilde{S}\in\F^{d\times d}$ such that $\tilde{U}\tilde{L} \tilde{S}=
L' \oplus I_r$, for some positive integer $r<d$, such that $L'$ is
right monic.

%% By scaling the last row appropriately, we can make $(d,d)^{th}$ entry
%% of $U_1'U_1LS_1$ equal to $1$ where $U_1'$ is the matrix associated
%% with operation of scaling the last row by factor
%% $\frac{1}{\alpha}$. Thereafter applying suitable row operations $R$
%% and column operations $S'$ on $U_1'U_1LS_1$ we make all entries of
%% $d^{th}$ row and $d^{th}$ column except $(d,d)^{th}$ entry
%% zero. I.e. we have $RU_1'U_1LS_1S' = \tilde{L}\oplus 1$ for some
%% linear matrix $\tilde{L}$. Crucially note that, as the last row of
%% $U_1L$ was made up of scalars, the matrix $S'$ is invertible
%% \emph{scalar} matrix. Clearly, $R$ is linear matrix unit.

%% If $\tilde{L}$ is right monic then we already have desired
%% matrices $U''=RU_1'U_1$, $S''=S_1S'$ such that linear matrix $U''$ is
%% unit, $S''$ is scalar invertible matrix such that $U''LS'' = \tilde{L}
%% \oplus 1$ and $\tilde{L}$ is right monic.

%% If $\tilde{L}$ is not right monic, we recursively apply the procedure
%% described above on $\tilde{L}$ to obtain unit $\tilde{U}$ and scalar
%% invertible matrix $\tilde{S}$ such that $\tilde{U}\tilde{L} \tilde{S}=
%% L' \oplus I_r$ for some $r$ such that $L'$ is right monic.

To see why this recursive procedure terminates for $r<d$, note that
the dimension of matrix $\tilde{L}$ is reducing by $1$ in each
recursive step and the matrix $\tilde{L}$ obtained is a stable
associate of $L$. So, if $r=d$ it would imply $L$ is a unit which is a
contradiction as we know that $L$ is obtained via Higman linearization
on a non-constant polynomial $f$, so $L$ is noninvertible.

Putting $U=RU_1\tilde{U}$ and $S=S_1\tilde{S}$ we have $ULS=L'\oplus
I_r$ where $L'$ is right monic as desired. It is clear that the entire
construction is polynomial time bounded, and that we have small ABPs
for the entries of $U$.
\end{proof}

\begin{remark}\label{remark-full-monic}
By repeated application of the algorithm in Theorem \ref{full-monic}
we can compute units $U_1,U_2\in\fR^{d\times d}$ such that $U_1 L U_2
= L' \oplus I_r$, where $L'$ is \emph{both left and right monic}. Such
a two-sided monic $L'$ is called \emph{monic} in \cite{Cohnfir}.

For our factorization algorithm, it suffices to compute an $L'$ that
is either left or right monic that is associated to $L$ as in Theorem
\ref{full-monic}. It turns out that either a left monic or a right
monic $L'$ suffices to use Ronyai's common invariant subspace
algorithm to factorize $L'$ (and hence also $L$) as we show in Theorem
\ref{lfact1}. More importantly, the fact that matrices $S$ and $S'$ in
Theorem~\ref{full-monic} are scalar is important for the factor
extraction algorithm as discussed in Theorem \ref{cnzthm}.
\end{remark}

\begin{lemma}\label{stable-algo}
  Given polynomials $f,g\in\fR$ as input by noncommutative arithmetic
  formulas, we can check in randomized polynomial time if $f$ and $g$
  are stable associates.
\end{lemma}

\begin{proof}
  Given $f$ and $g$, using Higman linearization we first compute in
  polynomial time full and monic linear matrices $A$ and $B$ such that
  $f$ and $A$ are stable associates and $g$ and $B$ are stable
  associates (see Theorem \ref{full-monic} and Remark
  \ref{remark-full-monic}). Now, $f$ and $g$ are stable associates iff
  $A$ and $B$ are stable associates. As both $A$ and $B$ are full and
  monic linear matrices, they are stable associates iff both $A$ and
  $B$ are matrices of the same dimension, say $d$, and there are
  scalar invertible matrices $P$ and $Q$ in $\F^{d\times d}$ such that
  $PA=BQ$ \cite[Theorem 5.8.3]{Cohnfir}, where $\F=\F_q$ or a small field extension. Letting the
  $2d^2$ entries of $P$ and $Q$ be variables, we can find a linearly
  independent set of solutions to $PA=BQ$ in polynomial time. Now,
  there exists invertible $P$ and $Q$ in the solution set iff the
  degree-$2d$ polynomial $\det P \times \det Q$ is nonzero on the
  solutions to $PA=BQ$. We can check this by the
  DeMillo-Lipton-Schwartz-Zippel Lemma \cite{DL78,Sch80,Zip79} by
  evaluating $\det P$ and $\det Q$ on a random linear combination of
  the basis of solutions to $PA=BQ$. This will be correct with high
  probablity.
\end{proof}

The next result shows how irreducibility (more generally, the property
of being an atom) is preserved by Higman linearization.

\begin{theorem}\label{associates-preseve-atoms}
Let $f\in\fR$ be a nonconstant polynomial and $L$ 
%$\in\fR^{\ell \times
%  \ell}$ 
  be a full linear matrix stably associated with $f$ (obtained via
Higman linearization).Then the
polynomial $f$ is irreducible iff $L$ is an atom.
 %by Theorem \ref{full-monic}). 
\end{theorem}

We give a self-contained proof of the above theorem, using the
following (suitably paraphrased) result of Cohn.

\begin{lemma}[Matrix Product Trivialization]{\rm\cite[pp.~198]{Cohnintro}}\label{triv-lemma}
    Let $A\in\fR^{m\times n}$ and $B\in\fR^{n\times s}$ be polynomial
    matrices such that their product $AB=0$. Then there exists a unit
    $P\in \fR^{n\times n}$ such that for every index $i\in [n]$ either
    the $i^{th}$ column of the matrix product $AP$ is all zeros or the
    $i^{th}$ row of the matrix product $P^{-1}B$ is all zeros.
\end{lemma}    

\begin{proofof}{Theorem~\ref{associates-preseve-atoms}}
By Higman linearization, we have upper and lower triangular matrices
$P$ and $Q$, respectively, such that
\[
f\oplus I_s = PLQ,
\]
for some positive integer $s$.

Now, if $f$ is not irreducible then we can write $f=f_1f_2$, where
$f_1$ and $f_2$ are both nonconstant polynomials in $\fR$. Hence
$f\oplus I_s$ factorizes as the product of non-units $(f_1\oplus
I_s)\cdot (f_2\oplus I_s)$, which implies the factorization
\[
L = P^{-1}(f_1\oplus I_s)(f_2\oplus I_s)Q^{-1}.
\]
Now, we claim $P^{-1}(f_1\oplus I_s)$ and $(f_2\oplus I_s)Q^{-1}$ are
non-units. Suppose $P^{-1}(f_1\oplus I_s)$ is a unit. Then $f_1\oplus
I_s$ is a unit which would imply there is an invertible matrix
$M\in\fR^{d\times d}$ such that $(f_1\oplus I_s)M=I_{s+1}$. But that
implies $f\cdot M_{1,1}=1$ which is impossible since $f_1$ is a
nonconstant polynomial. Similarly, $(f_2\oplus I_s)Q^{-1}$ cannot be a
unit. Hence $L$ is not an atom.

Conversely, suppose $L$ is not an atom. Then we can factorize it as
$L=M_1M_2$, where $M_1,M_2\in \fR^{d\times d}$ are
full non-units. Therefore, we have the factorization
\[
f\oplus I_s = (PM_1)(M_2Q).
\]

Writing the matrices $PM_1$ and $M_2Q$ as $2\times 2$ block matrices,
we have:
\[\left(
\begin{array}{c|c}
f & 0 \\
\hline
0 & I_s
\end{array}
\right) 
= \left(
\begin{array}{c|c}
c_1 & c_3 \\
\hline
c_2 & c_4
\end{array}
\right)\cdot 
\left(
\begin{array}{c|c}
d_1 & d_3 \\
\hline
d_2 & d_4
\end{array}
\right).
\]
{From} the $(2,1)^{th}$ matrix block on the left hand side of the
above equation, we obtain the following matrix identity:
\[
0 = \left(
\begin{array}{c c}
c_2 & c_4
\end{array}
\right)\cdot 
\left(
\begin{array}{c}
d_1 \\
d_2
\end{array}
\right),
\]
where $C=(c_2~c_4)$ is in $\fR^{s\times (s+1)}$ and $D=\left(
\begin{array}{c}
d_1 \\
d_2
\end{array}
\right)$ is in $\fR^{(s+1)\times 1}$. By Lemma~\ref{triv-lemma} there
is a unit $U\in \fR^{(s+1)\times (s+1)}$ such that for every $1\le
i\le s+1$ either the $i^{th}$ column of $C''=C\cdot U$ is all zeros or
the $i^{th}$ row of $D''=U^{-1} D$ is all zeros.  Note that $D''$, and
hence $D$, cannot be the all zeros column as $M_2Q$ is full. So, at
least one entry of $D''$ is nonzero. Hence, at least one column of
$C''$ is all zeros. By a suitable column permutation matrix $\Pi$ we
can ensure that the first column of $C \cdot U \Pi$ is all zeros.
Clearly, first entry of $ \Pi^{-1}U^{-1} D$ is nonzero. Writing $f
\oplus I_s$ as a product of $C'= PM_1U\Pi$ and $D'=\Pi^{-1}U^{-1}M_2Q$
we have
\[\left(
\begin{array}{c|c}
f & 0 \\
\hline
0 & I_s
\end{array}
\right) 
= \left(
\begin{array}{c|c}
c'_1 & c'_3 \\
\hline
c'_2 & c'_4
\end{array}
\right)\cdot
\left(
\begin{array}{c|c}
d'_1 & d'_3 \\
\hline
d'_2 & d'_4
\end{array}
\right),
\]

where $c_2'$ is an all zeros column and $d_1'$ is nonzero. {From} the
$(2,2)^{th}$ matrix block of the above equation, we obtain $c_4' d_4'
= I_s$ so $c_4'$ and $d_4'$ are units. By observing $(2,1)^{th}$
matrix block of the above equation we get $c_4' d_2' =0$, which
implies $d'_2$ is an all zeros column as $c_4'$ is unit.  It follows
that $f=c'_1\cdot d'_1$. Furthermore, it is a nontrivial factorization
because both $c'_1$ and $d'_1$ are non-units (because $C'$ and $D'$
are non-units, and $c'_4$ and $d'_4$ are units).
\end{proofof}

%% \begin{remark}
%% HKV18 cites Cohn prop. 0.5.2 and 3.2.1, section 3.2 for the proof of
%% above statement. It appears that we can give self contained proof
%% using existential trivialization for polynomial matrices.  $M\oplus
%% I_r= P(L'\oplus I_s)Q$. Now $M$ is non atom implies $M \oplus I_r$ is
%% non atom which implies $(L'\oplus I_s)$ is non atom. From here we want
%% to infer that $L'$ is non atom. We write the factorization of
%% $L'\oplus I_s$ as $2\times 2$ block matrices:
%% \[\left(
%% \begin{array}{c|c}
%% L' & 0 \\
%% \hline
%% 0 & I_{s}
%% \end{array}
%% \right) 
%% = \left(
%% \begin{array}{c|c}
%% c_1 & c_3 \\
%% \hline
%% c_2 & c_4
%% \end{array}
%% \right)\cdot 
%% \left(
%% \begin{array}{c|c}
%% d_1 & d_3 \\
%% \hline
%% d_2 & d_4
%% \end{array}
%% \right)
%% \]
%% Using idea of existential trivialization we can get non-trivial
%% factors of $L'$ from this equation. Converse direction follows
%% similarly.
%% \end{remark}

Let $L \in\fR^{d\times d}$ be a full and right (or left) monic linear matrix. Let
$L=A_0+\sum_{i=1}^n A_i x_i$. For a positive integer $\ell$ let $M_i,
i\in [n]$ be $\ell\times \ell$ scalar matrices with entries from $\F$
(or a small degree extension of $\F$). Let $Y_i, i\in [n]$ be
$\ell\times \ell$ matrices whose entries are distinct noncommuting
variables $y_{ijk}, 1\le j, k\le \ell$. Then the evaluation of the
linear matrix $L$ at $x_i\leftarrow Y_i+M_i, 1\le i\le n$ is the
$d\ell\times d\ell$ linear matrix in the $y_{ijk}$ variables:
\[
L' = A_0\otimes I_\ell + \sum_{i=1}^n A_i\otimes M_i +
\sum_{i=1}^n\sum_{j,k=1}^\ell (A_i\otimes E_{jk})\cdot y_{ijk}
\]

\begin{lemma}\label{shift-inv-lemma}
  There is a positive integer $\ell\le 2d$ such that for randomly
  picked $\ell \times \ell$ matrices $M_i, i\in[n]$ ( with entries
  from $\F$ or a small degree extension field) the matrix
  $A_0\otimes I_\ell + \sum_{i=1}^n A_i\otimes M_i$ is an
  invertible matrix.
\end{lemma}

\begin{proof}
Since $L\in\fR^{d\times d}$ is a full linear matrix, it has
noncommutative rank $d$. Hence, by the result of \cite{DM17} for the
generic $2d\times 2d$ matrix substitution $x_i \leftarrow X_i,
i\in[n]$, where $X_i$ is a matrix of distinct commuting variables, the
commutative rank of $L(X_1,X_2,\ldots,X_n)$ is $2d^2$ (which means it
is invertible). Hence there is a least $\ell\le 2d$ such that the
commutative rank of $L(X_1,X_2,\ldots,X_n)$ is $d\ell$, where $X_i$
are generic $\ell\times \ell$ matrices with commuting variables.
Hence, by the DeMillo-Lipton-Schwarz-Zippel lemma \cite{DL78,Sch80,Zip79} the
rank of the scalar matrix $L(M_1,M_2,\ldots,M_n)$ is $d\ell$, where
$M_i$ is a random scalar matrix with entries from $\F$ or a small
extension.
\end{proof}

Finally, we state and prove a \emph{modified version} of a result due
to Cohn that allows us to relate the factorization of a polynomial
$f\in\fR$ to the factorization of its Higman linearization $L$. The
proof is given in the appendix.

\begin{theorem}{\rm\cite[Theorem 5.8.8]{Cohnfir}}\label{cohnthm}
  Let $C\in\fR^{d\times d}$ be a full and right monic (or left monic)
  linear matrix for $d>1$. Then $C$ is not an atom if and only if
  there are $d\times d$ invertible scalar matrices $S$ and $S'$ such
  that
  \begin{equation}\label{eq-cohnthm}
  SCS' =
  \left(
\begin{array}{cc}
A & 0 \\
D & B
\end{array}
\right)
\end{equation}
where $A$ is an $r\times r$ full right (respec. left) monic linear matrix and $B$ is
an $s\times s$ full right (respec. left) monic linear matrix such that $r+s=d$.
\end{theorem}

\begin{remark}
  In \cite{Cohnfir} the theorem is proved under the stronger
  assumption that $C$ is monic. However, as we show, it holds even for
  $C$ that is right monic or left monic with minor changes to Cohn's
  proof. We require the above version for our factorization algorithm.
\end{remark}

%%   We can comment here about underlying field. Though the theorems
%% in Cohn are stated for algebraically closed fields, the ones we
%% need are true for any field, most importantly we need to state it
%% for 5.8.3 and 5.8.8. ...  \end{remark}

\section{Polynomial factorization: commutatively non-zero case}\label{cnz-sec}

Recall that $\fR$ denotes the free noncommutative polynomial ring
$\F\angle{x_1,x_2,\ldots,x_n}$ and our goal is to give a randomized
polynomial-time factorization algorithm for input polynomials in $\fR$
given as arithmetic formulas when $\F=\F_q$ is a finite field of size
$q$.

A polynomial $f\in\fR$ is \emph{commutatively nonzero} if
$f(\alpha_1,\alpha_2,\ldots,\alpha_n)\ne 0$ for scalars $\alpha_i\in
\F$ (or a small extension field of $\F$).

In this section we will present the factorization algorithm for
commutatively nonzero polynomials.\footnote{In the next section we
  will deal with the general case. The algorithm is more technical in
  detail, although in essence the same.}  It has three broad steps:
\begin{itemize}
\item[(i)] We transform the given polynomial $f$ to a full and \emph{right (or left)
  monic} linear matrix $L$ by first the Higman linearization of $f$
  followed by the algorithm in the proof of Theorem~\ref{full-monic}.
\item[(ii)] Next, we factorize the full and right (or left) monic linear matrix $L$
  into atoms.
\item[(iii)] Finally, we recover the irreducible factors of $f$ from
  the atomic factors of $L$.
\end{itemize}  

We formally state the three problems of interest in this paper.

\begin{problem}[$\prob{FACT}(\F)$]\hfill{~}\\
\textbf{Input}~A noncommutative polynomial $f\in\FX$ given by an
arithmetic formula.\\
\textbf{Output}~Compute a factorization
$f=f_1f_2\cdots f_r$, where each $f_i$ is irreducible, and each $f_i$
is output as an algebraic branching program.
\end{problem}

\begin{problem}[$\prob{LIN{-}FACT}(\F)$]\hfill{~}\\
\textbf{Input}~A full and right (or left) monic linear matrix $L\in\fR^{d\times 
  d}$.\\  
\textbf{Output}~Compute a factorization
$L=F_1F_2\cdots F_r$, where each $F_i$ is a full linear matrix that is
an atom.
\end{problem}

\begin{problem}[$\prob{INV}(\F)$]\hfill{~}\\
\textbf{Input}~A list of scalar matrices $A_1,A_2,\ldots,A_n\in
\F^{d\times d}$.\\
\textbf{Output}~Compute a nontrivial invariant
subspace $V\subset \F^d$ or report that the only invariant subspaces
are $0$ and $\F^d$.
\end{problem}  

In the three-step outline of the algorithm, for the second step we
will show that factoring a full and right (or left) 
monic linear matrix is randomized polynomial-time reducible to the
problem of computing a common invariant subspace for a collection of
scalar matrices. For the third step, we will give a polynomial-time
algorithm (based on Lemma~\ref{triv-lemma}) to recover the irreducible
factors of $f$ from the atomic factors of $L$.

\begin{remark}
  We use Ronyai's randomized polynomial-time algorithm \cite{Ronyai2}
  to solve the problem of computing a a common invariant subspace for
  a collection of matrices over $\F_q$.  Over rational numbers $\Q$,
  even for a special case the problem of computing a common invariant
  subspace turns out to be at least as hard as factoring square-free
  integers~\cite{FR85}. Hence, our approach to noncommutative
  polynomial factorization does not yield an efficient algorithm over
  $\Q$.
\end{remark}

Suppose $f\in\fR$ is given by a noncommutative arithmetic
formula. Since $f$ has small degree we can check if it is
commutatively nonzero in randomized polynomial-time by the
DeMillo-Lipton-Schwatrtz-Zippel test \cite{DL78,Sch80,Zip79} and, if
so, find $\alpha_i\in\F, i\in[n]$ such that
$f(\alpha_1,\alpha_2,\ldots,\alpha_n)\ne 0$ (if $\F$ is small, we pick
$\alpha_i$ from a small extension field). Furthermore, by a linear
shift of the variables $x_i\leftarrow x_i + \alpha_i, i\in[n]$
followed by scaling we can assume $f(\overline{0}) = 1$. Note that
from the factorization of the linear shift of $f$ we can recover the
factors of $f$ by shifting the variables back, and irreducibility is
preserved by linear shift. For the rest of this section we will
assume $f(\overline{0})=1$. 

Let $L=A_0+\sum_{i=1}^n A_i x_i$. As $f(\overline{0})=1$, we have
$L(\overline{0})=A_0$ is an invertible matrix. We now present an
efficient algorithm for factoring $L$ as a product of linear matrices
$L_1L_2\cdots L_r$, where each $L_i$ is an atom.

\begin{remark}
The factorization algorithm for arbitrary full and right (or left)
monic linear matrices (in which $A_0$ need not be invertible) is
similar but more involved. It is based on Lemma~\ref{shift-inv-lemma}
and is dealt with in the next section.
\end{remark}

\subsection{Algorithm for a special case of $\prob{LIN{-}FACT}(\F_q)$}

%%Let $L=A_0+\sum_{i=1}^n A_i x_i$ in $\fR^{d\times d}$ be a full, right
%%monic linear matrix such that $A_0$ is invertible and $\F=\F_q$.
%%In this special case we give an algorithm to factorize $L$ into a
%%product of atomic linear matrices in randomized polynomial time when
%%$\F=\F_q$ is a finite field.

\begin{theorem}\label{lfact1}
  There is a randomized polynomial-time algorithm for the following
  two special cases of the $\prob{LIN{-}FACT}(\F_q)$ problem:
  \begin{enumerate}
 \item Given a full right monic matrix $L$ as input such that
   $L(\overline{0})$ is an invertible matrix, the algorithm outputs a
   factorization of $L$ as a product of linear matrices that are
   atoms.
  \item Given a full left monic matrix $L$ as input such that
   $L(\overline{0})$ is an invertible matrix, the algorithm outputs a
   factorization of $L$ as a product of linear matrices that are
   atoms.  
\end{enumerate}
\end{theorem}  

\begin{proof}
  We present the algorithm only for the first part, as the second part
  has essentially the same solution.
  
  Let $L=A_0+\sum_{i=1}^n A_i x_i$ in $\fR^{d\times d}$ be such an
  instance of $\prob{LIN{-}FACT}(\F_q)$. We can write $L=A_0\cdot
  L'$ where $L'$ is the full and right monic linear matrix
  \[
  L'= I_d + \sum_{i=1}^n A_0^{-1}A_i x_i.
  \]
 Clearly, it suffices to factorize the linear matrix $L'$ into
 atoms. 
 
First we show that $L'$ is an not atom iff matrices $A_0^{-1}A_i$, $1\leq i \leq n$ have a nontrivial common invariant subspace.  
 By Theorem~\ref{cohnthm}, $L'$ is not an atom if and only if
 we can write $S_1L'S_2 = \left(
\begin{array}{cc}
B & 0 \\
D & C
\end{array}
\right)$ for invertible scalar matrices $S_1$ and $S_2$, where $B$ and
$C$ are full and right monic linear matrices, and $D$ is some linear
matrix. Equating the constant terms on both sides of 
the above equation we have $S_1S_2= \left(
\begin{array}{cc}
  B_0 & 0 \\
  D_0 & C_0
\end{array}
\right)$ as the constant term of $L'$ is $I_d$.  Thus the matrices $S_1S_2$ and its inverse also has the
same block form which implies that
$S_1L'S_1^{-1}=S_1L'S_2(S_1S_2)^{-1}$ also has the same block form. It follows that the $n$ matrices $A_0^{-1}A_i, 1\le i\le n$ have
a nontrivial common invariant subspace. Conversely, if the matrices
$A_0^{-1}A_i, 1\le i\le n$ have a nontrivial common invariant subspace
then we have a basic change scalar matrix $S$ such that $SL'S^{-1}$ has
the block form $\left(
\begin{array}{cc}
  L_1 & 0 \\
  * & L_2
\end{array}
\right)$,
where $L_1$ and $L_2$ are full and right monic linear matrices. So by Theorem~\ref{cohnthm} $L'$ is not an atom. 
 So we have established, $L'$  (and hence $L$) is not an atom iff matrices $A_0^{-1}A_i$, $1\leq i \leq n$ have a nontrivial common invariant subspace.
We will use Ronyai's randomized polynomial-time algorithm for finding
a nontrivial common invariant subspace for matrices $A_0^{-1}A_i, 1\le
i\le n$ over finite field $\F_q$. 

If there is no nontrivial invariant
subspace then the linear matrix $L'$ (and hence $L$) is an
atom. 
Otherwise, by repeated application of Ronyai's algorithm we will
obtain a basis change scalar matrix $T$ which when applied to $L'$
yields a linear matrix in the following \emph{atomic block diagonal
  form}:
\begin{equation}\label{eq3}
TL' T^{-1} = 
\left(
\begin{array}{ccccc}
L_1 & 0 &0 &\ldots & 0\\
* & L_2 &0 &\ldots & 0\\
* & * &L_3 &\ldots & 0\\
& & &\ddots &\\
* & * &* &\ldots & L_r\\
\end{array}
\right),
\end{equation}
where for each $j\in[r]$, the full right monic linear matrix $L_j \in \fR^{d_j \times
  d_j}$ is an atom, and each $*$ stands for some unspecified linear
matrix. It is now easy to factorize $TL'T^{-1}$ as a product of atoms
by noting one step of the factorization of $TL'T^{-1}$ from its form:
\[
TL'T^{-1}=  \left(
\begin{array}{cc}
A & 0 \\
D & L_r
\end{array}
\right) =   \left(
\begin{array}{cc}
A & 0 \\
0 & I
\end{array}
\right)\cdot 
  \left(
\begin{array}{cc}
I & 0 \\
D & I
\end{array}
\right)\cdot
\left(
\begin{array}{cc}
I & 0 \\
0 & L_r
\end{array}
\right).
\]
We note that  $\left(
\begin{array}{cc}
I & 0 \\
D & I
\end{array}
\right)$ is a unit. Since $L_r$ is an atom the product $\left(
\begin{array}{cc}
I & 0 \\
D & I
\end{array}
\right)\cdot \left(
\begin{array}{cc}
I & 0 \\
0 & L_r
\end{array}
\right)$ is also an atom and a linear matrix, and it is the rightmost
factor of $TL'T^{-1}$. Continuing thus with $A$ now, we can factorize
$TL'T^{-1}$ as a product $F'_1F'_2\cdots F'_r$ of $r$ atoms, each of
which is a linear matrix. It follows that $L=A_0T^{-1}F'_1F'_2\cdots
F'_rT$ is a complete factorization of $L$ as a product of atomic
linear matrices (both $A_0$ and $T$ are scalar invertible matrices).
\end{proof}

\begin{remark}
We note that Ronyai's algorithm \cite{Ronyai2} for $\prob{INV}(\F_q)$
is actually a deterministic polynomial-time \emph{reduction} from
$\prob{INV}(\F_q)$ to univariate polynomial factorization over $\F_q$.
\end{remark}

%Below we will be describing our algorithm for the
%$\prob{LIN{-}FACT}(\F_q)$ where $L$ is full and right or left monic
%matrix. 

Based on whether we want to work with right monic or left
monic case we will express $f \oplus I_s$ in an appropriate form using
Higman linearization and Theorem~\ref{full-monic} as described in the
equation below:
\begin{equation}\label{eq2}
f\oplus I_s = 
\begin{cases}
~~PU(L'\oplus I_t) SQ,~ \text{in the right monic case}\\ 
~~PS(L'\oplus I_{t}) UQ,~ \text{in the left monic case} 
\end{cases}
\end{equation}

where $d+t=s+1$, $L'\in\fR^{d\times d}$ is a full and right (or left) monic linear
matrix, $P$ is upper triangular with all $1$'s diagonal, $Q$ is lower
triangular with all $1$'s diagonal, $U\in\fR^{(d+t)\times (d+t)}$ is a
unit, and $S\in\F^{(d+t)\times (d+t)}$ is an invertible scalar
matrix.

\subsection*{Algorithm for $\prob{FACT}(\F_q)$}

We are now ready to describe the polynomial factorization algorithm
for commutatively nonzero polynomials in $\fR$. Starting with the
Higman linearization of the input polynomial $f\in\fR$ as in
Equation~\ref{eq2}, by an application of the first parts of Theorems
\ref{full-monic} and \ref{lfact1} we obtain the factorization
$f\oplus I_s = PUF'_1F'_2\cdots F'_r SQ$ using the structure in
Equation~\ref{eq3}. 

Alternatively, by applying the second part of Theorem \ref{full-monic}
we can compute a left monic linear matrix $L'$ that is a stable
associate of $f$ and, applying the second part of Theorem~\ref{lfact1}
we can compute the factorization
\begin{equation}\label{eq4}
f\oplus I_s = PS'F'_1F'_2\cdots F'_r U'Q.
\end{equation}
where each linear matrix $F'_i$ is an atom, $P$ is upper triangular
with all $1$'s diagonal, $Q$ is lower triangular with all $1$'s
diagonal, $U'$ is a unit and $S'$ is a scalar invertible matrix.
Equation~\ref{eq4} is the form we will use for the algorithm (we could
equally well use the other factorization).

{From} the structure of the atomic block diagonal matrix $TL'T^{-1}$
in Equation~\ref{eq3} notice that the product $S'F'_1F'_2\cdots F'_i $
is a linear matrix for each $1\le i<r$.

The next lemma presents an algorithm that is crucial for extracting
the factors of $f$.

\begin{lemma}\label{triv-lem}
  Let $C\in\fR^{u\times d}$ be a linear matrix and $v\in\fR^{d\times
    1}$ be a column of polynomials such that $Cv=0$. Each entry $v_i$ of $v$ is given by an
  algebraic branching program as input. Then, in polynomial time we
  can compute a invertible matrix $N\in\fR^{d\times d}$ such that
  \begin{itemize}
  \item For $1\le i\le d$ either the $i^{th}$ column of $CN$ is all zeros
    or the $i^{th}$ row of $N^{-1}v$ is zero.
  \item Each entry of $N$ is a polynomial of degree at most $d^2$ and
    is computed by a polynomial size ABP, and also each entry of
    $N^{-1}$ is computed by a polynomial size ABP.
  \end{itemize}
%%  $N$ is lower triangular with all $1$'s diagonal.  
\end{lemma}

\begin{proof}
 We will describe the algorithm as a recursive procedure Trivialize
 that takes matrix $C$ and column vector $v$ as parameters and returns
 a matrix $N$ as claimed in the statement.

  \begin{enumerate}
  \item[] Procedure Trivialize$(C\in\fR^{u\times d},v\in\fR^{d\times 1})$
  \item If $d=1$ then (since $Cv=0$ iff either $C=0$ or $v=0$)
    \textbf{return} the identity matrix.
  \item  If $d>1$ then
  \item write $C=C_0+C_1$, where $C_0$ is a scalar matrix and $C_1$ is the
    degree $1$ homogeneous part of $C$. Let $k$ be the degree of the
    highest degree nonzero monomials in the polynomial vector $v$, and
    let $m$ be a nonzero degree $k$ monomial. Let
    $v(m)\in\F_q^{d\times 1}$ denote its (nonzero) coefficient in $v$.
    Then $Cv=0$ imples $C_1 v(m)=0$. Let $T_0\in\F_q^{d\times d}$ be a
    scalar invertible matrix with first column $v(m)$ obtained by
    completing the basis.
    \begin{enumerate}
    \item If $C_0v(m)=0$ then the first column of $CT_0$ is zero.
    \item Otherwise, $CT_0$ has first column as the nonzero scalar
      vector $Cv(m)=C_0v(m)$. Suppose $i^{th}$ entry of $Cv(m)$ is a
      nonzero scalar $\alpha$. With column operations we can drive the
      $i^{th}$ entry in all other columns of $CT_0$ to zero. Let the
      resulting matrix be $CT_0T_1$ (where the matrix $T_1$ is
      invertible as it is a product of elementary matrices
      corresponding to these column operations, each of which is of
      the form $\col_i\leftarrow (\col_i +
      \col_1\cdot \alpha_0+\sum_i \alpha_i x_i)$). Notice that $CT_0T_1$ is still
      linear.
    \item As $Cv=(CT_0T_1)(T_1^{-1}T_0^{-1}v)$, and in the $i^{th}$
      row of $CT_0T_1$ the only nonzero entry is $\alpha$ which is in
      its first column, we have that the first entry of
      $T_1^{-1}T_0^{-1}v$ is zero.
    \end{enumerate}
  \item Let $C'\in\fR^{u\times (d-1)}$ obtained by dropping the first
    column of $CT_0T_1$. Let $v'\in\fR^{(d-1)\times 1}$ be obtained by
    dropping the first entry of $T_1^{-1}T_0^{-1}v$. Note that $C'$
    is still linear.
  \item Recursively call Trivialize$(C'\in\fR^{u\times
    (d-1)},v'\in\fR^{(d-1)\times 1})$.  and let the matrix returned by
    the call be $T_2\in\fR^{(d-1)\times (d-1)}$.
  \item Putting it together, return the matrix $T_0T_1(I_1\oplus T_2)$.
  \end{enumerate}

  To complete the proof, we note that a highest degree monomial $m$
  such that $v(m)\ne 0$ is easy to compute in deterministic polynomial
  time if each $v_i$ is given by an algebraic branching program using
  the PIT algorithm of Raz and Shpilka \cite{RS05}. Notice that for
  the recursive call we need $C'$ to be also a linear matrix and each
  entry of $v'$ to have a small ABP. $C'$ is linear because $CT_0T_1$
  is a linear matrix since $CT_0$ is linear, its first column is
  scalar, and each column operation performed by $T_1$ is scaling the
  first column of $CT_0$ by a linear form and subtracting from another
  column of $CT_0$. Each entry of $v'$ has a small ABP because
  $T_0^{-1}$ is scalar and it is easy to see that the entries of
  $T_1^{-1}$ have ABPs of polynomial size. Finally, we note that $T_1$
  is a product of at most $d-1$ linear matrices (each corresponding to
  a column operation), and $N$ is an iterated product of $d$ such
  matrices. Hence, each entry of $N$ as well as $N^{-1}$ is a
  polynomial of degree at most $d^2$ and is computable by a small ABP.

\end{proof}

%%\begin{remark}
%%A more careful argument shows, in fact, that taking the linear part of
%%$N$ obtained in the above would also suffice.
%%\end{remark}

Turning back to our algorithm for $\prob{FACT}(\F_q)$, in the next
lemma we design an efficient algorithm that will allow us to extract
all the irreducible factors of $f$ (given Equation~\ref{eq4}).

\begin{lemma}[Factor Extraction]\label{extract}
  Let $f\in\fR$ be a polynomial and $G\in\fR^{(d-1)\times (d-1)}$ be a
  unit such that
  \begin{equation}\label{eq5}
  \left(
\begin{array}{cc}
f & u \\
0 & G
\end{array}
\right)
= P CD,
\end{equation}
such that
\begin{itemize}
\item $C$ is a full linear matrix that is a non-unit, $P$ is upper
  triangular with all $1$'s diagonal, and $D\in\fR^{d\times d}$ is a
  full non-unit matrix which is also an atom.
\item The polynomial $f$, and the entries of $u, G, P, D$ are all
  given as input by algebraic branching programs.
\end{itemize}
Then we can compute in deterministic polynomial time a nontrivial
factorization $f=g\cdot h$ of the polynomial $f$ such that $h$ is an
irreducible polynomial.
\end{lemma}

\begin{proof}
Let
\[C=\left(
\begin{array}{cc}
c_1 & c_3 \\
c_2 & c_4
\end{array}
\right) \text{ and }
D = 
  \left(
\begin{array}{cc}
d_1 & d_3\\
d_2 & d_4
\end{array}
\right),
\]
written as $2\times 2$ block matrices where $c_1$ and $d_1$ are
$1\times 1$ blocks. By dropping the first row of the matrix in the
left hand side of Equation~\ref{eq5} and the first row of $P$ we
get
\[
  (0~G) = (0~P') C D,
\]
where $P'$ is also an upper triangular matrix with all $1$'s diagonal.
Equating the first columns on both sides we have
\begin{eqnarray*}
  0 & =&  (0~P')\left(
\begin{array}{cc}
c_1 & c_3 \\
c_2 & c_4
\end{array}
\right)\left(\begin{array}{c} d_1\\ d_2 \end{array} \right),
  \text{ which implies that}\\
  0 & = & P'(c_2~c_4)\left(\begin{array}{c} d_1\\ d_2 \end{array} \right),
  \text{ and hence}\\
  0 & = & (c_2~c_4)\left(\begin{array}{c} d_1\\ d_2 \end{array} \right),
  \text{ since $P'$ is invertible.}
\end{eqnarray*}

%$\left(\begin{array}{c} d_1\\ d_2 \end{array} \right)\in\fR^{d\times 1}$

Since $(c_2~c_4)\in\fR^{(d-1)\times d}$ is a matrix with linear
entries and $\left(\begin{array}{c} d_1\\ d_2 \end{array} \right)\in\fR^{d\times 1}$ is a column vector of
polynomials which are given by ABPs as input, we can apply the
algorithm of Lemma~\ref{triv-lem} to compute a unit
$N$ such that its entries are all given by
ABPs such that for $1\le i\le d$, either the $i^{th}$ column of
$(c'_2~c'_4)=(c_2~c_4)N$ is zero or the $i^{th}$ row of
$\left(\begin{array}{c} d_1'\\ d_2' \end{array}
 \right)=N^{-1}\left(\begin{array}{c} d_1\\ d_2 \end{array}
 \right)$ is zero.

%% $N^{-1}\left(\begin{array}{c} d_1\\ d_2 \end{array}
%% \right)=\left(\begin{array}{c} d'_1\\ d'_2 \end{array} \right)$ is
%% zero.

Now the following argument is almost identical with the argument towards the end of the proof of the Theorem \ref{associates-preseve-atoms}. We give it below for completeness. 
Since $D$ is a full matrix, the matrix $N^{-1}D$ is also full which
implies its first column $\left(\begin{array}{c} d_1'\\ d_2' \end{array} \right)$ cannot be all zeros. So there is at least one nonzero entry in $\left(\begin{array}{c} d_1'\\ d_2' \end{array} \right)$ and the corresponding column in $(c'_2~c'_4)$ is all zero. This implies there exist a permutation matrix $\Pi$ such that the first column of $C(c'_2~c'_4)\Pi$ is all zero and first entry of $\Pi^{-1} \left(\begin{array}{c} d_1'\\ d_2' \end{array} \right)$ is non zero.

%Since $C$ is a full matrix, $(c_2~c_4)$ has (noncommutative) rank
%$d-1$ which implies $(c'_2~c'_4)$ has rank $d-1$ as $N$ is invertible.
%Thus, $(c'_2~c'_4)$ has at most one column of all zeros.  Furthermore,
% Putting
%it together, it follows that there is a \emph{unique} $i_0\in[d]$ such
%the $i_0^{th}$ column of $(c'_2~c'_4)$ is zero and the $i_0^{th}$ row
%of $(d'_1~d'_2)^T$ is nonzero. Let $\Pi$ be the permutation matrix
%that permutes the $i_0^{th}$ column of $(c'_2~c'_4)$ to the first
%column. Then $\Pi^T=\Pi^{-1}$ permutes the $i_0^{th}$ row of
%$(d'_1~d'_2)^T$ to the first position. Therefore, letting
%$(c''_2~c''_4)=(c_2~c_4)N\Pi$ and
%$(d''_1~d''_2)^T=\Pi^{-1}N^{-1}(d_1~d_2)^T$, we have $c''_2=0$ and
%$d''_2=0$ and $d''_1\ne 0$. 

Consider the matrices $C''=CN\Pi=\left(
\begin{array}{cc}
c''_1 & c''_3 \\
c_2'' & c''_4
\end{array}
\right)
$ and
$D''=\Pi^{-1}N^{-1}D=\left(
\begin{array}{cc}
d''_1 & d''_3 \\
d_2'' & d''_4
\end{array}
\right)
$. We have

\[  \left(
\begin{array}{cc}
f & * \\
0 & G'
\end{array}
\right)
=P^{-1}
\left(
\begin{array}{cc}
f & u \\
0 & G
\end{array}
\right)
=
\left(
\begin{array}{cc}
c''_1 & c''_3 \\
c_2'' & c''_4
\end{array}
\right)
  \left(
\begin{array}{cc}
d''_1 & d''_3\\
d_2'' & d''_4
\end{array}
\right)
\]

, where $G'=(P')^{-1}G$ is a unit, $c_2''$ is all zero column matrix and $d_1''$ is non-zero. Now observing $(2,1)^{th}$ matrix block in the above equation, we get $d_2''$ is all zero column. Hence, by looking at $(2,2)^{th}$ block in the above equation, we can see that $c_4''$ and $d_4''$ are units as $G'$ is a unit. Clearly, we have $f= c_1'' \cdot d_1 ''$. Now,
since $C$ and $D$ are non-units (by assumption), the matrices $C''$
and $D''$ are also non-units. Therefore, $c''_1$ is not a scalar for
otherwise $C''$ would be a unit. Similarly, $d''_1$ is not a scalar.
It follows that $f=c''_1d''_1$ is a nontrivial factorization of $f$.

Furthermore, since $D$ is an atom by assumption and $D''$ is a stable
associate of $D$, $D''$ is an atom. As $D'' = \left(
\begin{array}{cc}
d''_1 & d''_3\\
0 & d''_4
\end{array}
\right)$ and $d_4''$ is invertible, we get $\left(
\begin{array}{cc}
1 & 0\\
0 & (d''_4)^{-1}
\end{array}
\right) \cdot D'' = \left(
\begin{array}{cc}
d''_1 & d''_3\\
0 & I_{s}
\end{array}
\right)$. Now applying suitable row operations to the matrix $(1 \oplus (d''_4)^{-1})D''$ we can drive $d_3''$ to zero. So we have $U(1 \oplus (d''_4)^{-1})D''=(d_1'' \oplus I_s)$ for a unit $U$. Hence $d_1''$ is an associate of $D''$ and therefore $d_1''$ is irreducible as $D''$ is an atom.

%\[
%  \left(
%\begin{array}{cc}
%f & u \\
%0 & G
%\end{array}
%\right) = P C''D'' =
%\left(
%\begin{array}{cc}
%1 & * \\
%0 & P'
%\end{array}
%\right)
%\left(
%\begin{array}{cc}
%c''_1 & c''_3 \\
%0 & c''_4
%\end{array}
%\right)
%  \left(
%\begin{array}{cc}
%d''_1 & d''_3\\
%0 & d''_4
%\end{array}
%\right),
%\]
%which yields, since $P$ is upper triangular with all $1$'s diagonal,
%that $f=c''_1 d''_1$ and $G=P'c''_4d''_4$. By assumption, the matrix
%$G$ is a unit. So, $G^{-1}P'c''_4d''_4=I_{d-1}$ which implies both
%$c''_4$ and $d''_4$ are also units in $\fR^{(d-1)\times (d-1)}$.  Now,
%since $C$ and $D$ are non-units (by assumption), the matrices $C''$
%and $D''$ are also non-units. Therefore, $c''_1$ is not a scalar for
%otherwise $C''$ would be a unit. Similarly, $d''_1$ is not a scalar.
%It follows that $f=c''_1d''_1$ is a nontrivial factorization of $f$.

%Furthermore, since $D$ is an atom by assumption and $D''$ is a stable
%associate of $D$, it follows that $d''_1$ is an irreducible
%polynomial.
\end{proof}

Finally, we describe the factorization algorithm for commutatively
nonzero polynomials $f\in\fR$ over finite fields $\F_q$.

\begin{theorem}\label{cnzthm}  
  Let $\fR=\F_q\angle{X}$ and $f\in \fR$ be a commutatively nonzero
  polynomial given by an algebraic branching program of size $s$ as
  input instance of $\prob{FACT}(\F_q)$. Then there is a
  $\poly(s, \log q)$ time randomized algorithm that outputs a
  factorization $f=f_1f_2\cdots f_r$ such that each $f_i$ is
  irreducible and is output as an algebraic branching program.
\end{theorem} 

\begin{proof}
  Given $f$ as input, we apply Higman linearization followed by the
  algorithm for $\prob{LIN{-}FACT}(\F_q)$ described in
  Theorem~\ref{lfact1} to obtain the factorization of $f\oplus I_s=PSS_1F_1F_2 \ldots F_rS_2UQ$ where 
each linear matrix $F_i$ is an atom, $P$ is upper triangular
with all $1$'s diagonal, $Q$ is lower triangular with all $1$'s
diagonal, $U$ is a unit and $S$ is a scalar invertible matrix, as given in Equation~\ref{eq4}. 
%Equation~\ref{eq4} is the form we will use for the algorithm
   We can now apply Lemma~\ref{extract} to
  extract irreducible factors of $f$ (one by one from the right).

  For the first step, let $C=S S_1F_1F_2\cdots F_{r-1}$ and $D=F_rS_2UQ$
  in Lemma~\ref{extract}. The proof of Lemma~\ref{extract} yields the
  matrix $N_r=N\Pi$ such that both matrices $C''=PS S_1F_1F_2\cdots
  F_{r-1}N_r$ and $D''=N_r^{-1}F_rS_2UQ$ has the first column all zeros
  except the $(1,1)^{th}$ entries $c''_1$ and $d''_1$ which yields the
  nontrivial factorization $f=c''_1d''_1$, where $d''_1=f_r$ is 
  irreducible. Renaming $c''_1$ as $g_r$ we have from the structure of
  $C''$:
  \[
  \left( \begin{array}{cc} g_r & * \\ 0 & G_r \end{array} \right) =
  P(SS_1F_1F_2\cdots F_{r-2}) (F_{r-1}N_r).
  \]
Setting $C= SS_1F_1F_2\cdots F_{r-2}$ and $D=F_{r-1}N_r$ in
Lemma~\ref{extract} we can compute the matrix $N_{r-1}$ using which we
will obtain the next factorization $g_r=g_{r-1}f_{r-1}$, where
$f_{r-1}$ is irreducible because the linear matrix $F_{r-1}$ is an
atom. Lemma~\ref{extract} is applicable as all conditions are met by
the matrices in the above equation (note that $G_r$ will be a unit).

Continuing thus, at the $i^{th}$ stage we will have
$f=g_{r-i+1}f_{r-i+1}f_{r-i+2}\cdots f_r$ after obtaining the
rightmost $i$ irreducible factors by the above process. At this stage
we will have
  \[
  \left( \begin{array}{cc} g_{r-i+1} & * \\ 0 & G_{r-i+1} \end{array}
  \right) = P(SS_1F_1F_2\cdots F_{r-i-1}) (F_{r-i}N_{r-i+1}),
  \]
  where $G_{r-i+1}$ is a unit and all other conditions are met to
  apply Lemma~\ref{extract}.

  Thus, after $r$ stages we will obtain the complete factorization
  $f=f_1f_2\cdots f_r$. For the running time, it suffices to note that
  the matrix $N$ computed in Lemma~\ref{extract} is a product of
  degree at most $d^2$ many linear matrices (corresponding to the
  column operations). Thus, at the $i^{th}$ of the above iteration,
  the sizes of the ABPs for the entries of $N_{r-i+1}$ are independent
  of the stages. Hence, the overall running time is easily seen to be
  polynomial in $s$ and $\log q$.
\end{proof}

\begin{corollary}\label{sparse-cor1}
If $f\in\fR$ is commutatively nonzero polynomial given as input in
sparse representation (as an $\F_q$-linear combination of its
monomials) then in randomized polynomial time we can compute
a factorization into irreducible factors in sparse representation.
\end{corollary}

\begin{proof}
  Let $f$ be given as input in sparse representation. Suppos $\deg f =
  d$ and it is $t$-sparse. Then there are at most $td^2$ many
  monomials that can occur as a substring of the monomials of $f$.  We
  can apply the randomized algorithm of Theorem~\ref{cnzthm} to obtain
  the factorization $f=f_1f_2\cdots f_r$, where each $f_i$ is given by
  an ABP. Now, for each of the $td^2$ many candidate monomials of
  $f_i$ we can find its coefficient in $f_i$ in polynomial time (using
  the Raz-Shpilka algorithm \cite{RS05}). Hence we can obtain the
  factorization $f=f'_1f'_2\cdots f'_r$, where each $f'_i$ is a
  $t$-sparse polynomial.
\end{proof}

%\section{statements of main theorems}
\section{Factorization of Commutatively zero polynomials}\label{cz-sec}

In this section we will describe the general case of the factorization
algorithm when the input polynomial $f\in\fR$ is a commutatively zero
polynomial. That is, $f$ evaluates to zero on all scalar substitutions
from $\F_q$ or any (commutative) extension field.

%%In Section \ref{cnz-sec} we obtained a randomized polynomial-time
%%factorization algorithm for commutatively non-zero polynomials in
%%$\fR$ given as arithmetic formulas.

The factorization algorithm will follow the three broad steps
described in Section \ref{cnz-sec} for the commutatively nonzero case:
first, using Higman linearization and Theorem \ref{full-monic},
transform the polynomial $f$ to a stably associated linear matrix $L$
that is full and left (or right) monic. Next, factorize the linear
matrix $L$ into atoms. Finally, recover the irreducible factors of $f$
from the atomic factors of the linear matrix $L$ using the factor
extraction procedure described in Lemma~\ref{extract}.

The step that requires a new algorithm is factorizing a full and right (or left)
monic linear matrix $L\in\fR$ into atoms when $f$ is commutatively
zero, which means there is no scalar substitution $x_i\leftarrow
\alpha_i, i\in[n]$ such that $L(\alpha_1,\alpha_2,\ldots,\alpha_n)$ is
invertible. Note that in this case we cannot apply the algorithm for
factorizing a linear matrix as discussed in the proof of Theorem
\ref{lfact1}).

\subsection{Factorization of full and monic linear matrices}

Let $f\in \fR$ be the input polynomial given by a size $s$ formula and
let $L \in\fR^{d\times d}$ be a full, right monic linear matrix stably
associated with $f$ obtained via Higman linearization and an
application of Theorem \ref{full-monic}.

Recall, by Equation~\ref{eq2} we have $f \oplus I_s= PU(L \oplus
I_t)SQ$ where, $P$, $Q$ are respectively upper triangular and lower
triangular units with diagonal entries $1$, $U$ is a unit and $S$ is
scalar invertible. 

%%Since $L$ is a full linear matrix, by Lemma~\ref{shift-inv-lemma} we
%%can find scalar matrix substitutions $x_i\leftarrow M_i,
%%M_i\in\F_q^{s\times s}$ for $s\le d$ such that the matrix
%%$L(M_1,M_2,\ldots,M_n)$ is invertible. Matrix substitutions in a
%%linear matrix has been used earlier in the context of noncommutative
%%rank, see e.g. \cite{IQS17}, \cite{DM15}, \cite{HKV20}.

Let $L=A_0+\sum_{i=1}^n A_n x_i\in\fR^{d\times d}$ be the given full
and right monic linear matrix.  First, by Lemma~\ref{shift-inv-lemma},
we will find a suitable scalar matrix $n$-tuple
$\bar{M}=(M_1,M_2,\ldots,M_n)$, each $M_i\in \F_q^{\ell\times \ell}$
for $\ell\le 2d$, such that under the substitution $x_i\leftarrow M_i$
the matrix $L(\bar{M})$ is invertible.

For $1\le i\le n$ let $Y_i$ be an $\ell\times \ell$ matrix of distinct
noncommuting variables $y_{ijk}$. We consider the dilated linear
matrix
\begin{equation}\label{dilate-eq1}
L' = A_0\otimes I_\ell + \sum_{i=1}^n A_i\otimes (Y_i+M_i).
\end{equation}
It is not hard to see that $L'$ is full and $L'$ is right monic as  $L$ is right monic. Additionally, its constant term is 
invertible. So, we can apply Theorem~\ref{lfact1} to factorize $L'$ as
a product of two linear matrices, both non-units.

The following lemma \cite{HKV20} has an important role in our
algorithm for recovering the factorization for $L$ from a
factorization of $L'$.

\begin{lemma}\cite{HKV20}\label{hkv20-main-lemma}

Let $L \in \fR^{d \times d}$ be a full linear matrix with $L= A_0+ A_1x_1+\ldots +A_n x_n$ such that $A_i \neq 0$ for at least one $i$, $1\leq i \leq n$ and $L'\in
R^{d\ell \times d\ell}$ be a matrix obtained from $L$ by substituting
variable $x_i$ by $Y_i$ for $i \in [n]$, where $Y_i$ is $\ell \times
\ell$ matrix whose $(j,k)^{th}$ entry is a fresh noncommuting variable
$y_{i,j,k}$ for $1 \leq j,k \leq \ell$. Then
\begin{enumerate}
\item If $L'$ is of the form $G L' H = \left(
\begin{array}{c|c}
A' & 0 \\
\hline
D' & B'
\end{array}
\right)$, where $A'$ is $d' \times d'$ matrix and $B'$ is $d'' \times d''$ matrix for $0< d', d''$, with $d'+d''=d\ell$ and $G,H$ are $d\ell \times d\ell$ invertible scalar matrices then there exist $d \times d$ invertible scalar matrices $U, V$ such that  $U L V = \left(
\begin{array}{c|c}
A & 0 \\
\hline
D & B
\end{array}
\right)$, where $A$ is $e' \times e'$ matrix and $B$ is $e'' \times e''$ matrix for $0<e', e''$, with $e' +e''= d$. 

\item Moreover, given $L'$ explicitly along with its representation
  mentioned above, we can find the matrices $U, V$ in deterministic
  polynomial time (in $n,\ell,d$).
\end{enumerate}
\end{lemma}

\begin{remark}
We give a self-contained complete proof of the above linear-algebraic
lemma in the appendix for $\F_q$, because the proof given in
\cite{HKV20} is sketchy in parts with some details missing, and also
their lemma is stated only for complex numbers and they are not
concerned about computing the matrices $U$ and $V$.
\end{remark}

Now, we can apply Lemma \ref{hkv20-main-lemma} to transform the
factorization of $L'$ to a factorization of $L$ as a product of two
linear matrices, both non-units. Repeating the above on both the
factors of $L$ we will get a complete atomic factorization of
$L$. Formally, we prove the following.

\begin{theorem}\label{lfact2}
On input a full and right (or left) monic linear matrix $L= A_0 +
\sum_{i=1}^n A_i x_i$ where $A_i \in \mathbb{F}^{d \times d}$ for $i\in
    [n]$, there is a randomized polynomial time ($poly(n,d)$)
    algorithm to compute scalar invertible matrices $S,S'$ such that
    $SLS'$ has atomic block diagonal form.
\end{theorem}

\begin{proof}
  We present the algorithm only for right monic $L$; the left monic
  case has essentially the same solution.

  If the input $L$ is not full or right monic the algorithm can
  efficiently detect that and output ``failure''. If $L$ is an atom
  the algorithm will output that $L$ is an atom and set the matrices
  $S$ and $S'$ to $I_d$. Otherwise, the algorithm will compute
  invertible scalar matrices $S$ and $S'$ such that
\begin{equation}\label{eq-block-atom}
S L S' = 
\left(
\begin{array}{ccccc}
L_1 & 0 &0 &\ldots & 0\\
* & L_2 &0 &\ldots & 0\\
* & * &L_3 &\ldots & 0\\
& & &\ddots &\\
* & * &* &\ldots & L_r\\
\end{array}
\right),
\end{equation}
where the matrix on the right is in atomic block diagonal form, that
is, each linear matrix $L_i$ is an atom.

\begin{enumerate}
\item[]   {\bf Procedure Factor(L)}.
\item Test if $L$ has full noncommutative rank using the algorithm in
  \cite{IQS17} or \cite{GGOW20}. Test if $L$ is right monic by checking if
  the matrix $[A_1 A_2 \ldots A_n]$ has full row rank (which is $d$).
  If $L$ is not full and right monic the algorithm outputs ``fail''.
\item Assume $L$ is full and right monic. Using Lemma
  \ref{shift-inv-lemma}, find smallest positive integer $\ell \leq 2d$
  and $\ell\times \ell$ scalar matrices $M_i, i\in [n]$ with entries
  from $\F$ (or a small degree extension of $\F$) such that
  $W=L(\bar{M})$ is $d \cdot \ell \times d \cdot \ell$ invertible
  scalar matrix. Compute the dilated linear matrix $L'$ in the
  $y_{ijk}$ variables as in Equation~\ref{dilate-eq1} which can be
  rewritten as:
\[
L' = A_0\otimes I_\ell + \sum_{i=1}^n A_i\otimes M_i +
\sum_{i=1}^n\sum_{j,k=1}^\ell (A_i\otimes E_{jk})\cdot y_{ijk}.
\]
Let $L''= W^{-1} L'$. Clearly $L''(\overline{0})= I_{d\ell}$. Hence,
by the algorithm of Theorem \ref{lfact1} we can either detect that
$L''$ is an atom or factorize $L''$. If $L''$ is an atom then $L$ is
also an atom and the algorithm can output that and stop.  Otherwise,
$L'$ is not an atom and by Theorem \ref{lfact1} we will obtain a basis  
change matrix $T$ such that $T W^{-1} L' T^{-1}= T L'' T^{-1} = \left(
\begin{array}{cc}
C'' & 0 \\
* & D''
\end{array}
\right)$ where $C''$ and $D''$ are linear matrices of dimension $c''
\times c''$ and $d'' \times d''$ respectively, such that $c''+ d'' = d
\ell$.

%Here we note that, $L$ is not an atom implies $L'$ is not an atom (as
%the factorization of $L'$ can be obtained by substituting
%$x_i\leftarrow Y_i+M_i, 1\le i\le n$ in the factors of $L$). The
%other direction, namely, if $L'$ is not an atom implies $L$ is not an
%atom, follows from Lemma \ref{}. So in this step algorithm correctly
%determines whether $L$ is an atom or not.
  
\item By linear shift of variables $y_{ijk} \leftarrow y_{ijk} -
  M_i(j,k)$ we obtain $\tilde{T} \tilde{L} \tilde{T'} = \left(
\begin{array}{cc}
C' & 0 \\
* & D'
\end{array}
\right)$    
for some scalar invertible matrices $\tilde{T}, \tilde{T'}$ where
$\tilde{L} = L( Y_1, \ldots, Y_n)$.

\item Applying the algorithm of Lemma \ref{hkv20-main-lemma} to
  $\tilde{L}$, $\tilde{T}$, and $\tilde{T'}$, in deterministic
  polynomial time we obtain scalar invertible matrices $\tilde{S}, \tilde{S'}$ such
  that $\tilde{S}L\tilde{S'} = \left(
\begin{array}{cc}
C & 0 \\
* & D
\end{array}
\right) $
where $C$, $D$ are square matrices of dimensions $e \times e$ and $g
\times g$, respectively, such that $e+g = d$.

\item Recursively call Factor$(C)$ and Factor$(D)$. Let $S_1, S'_1$ be
  the matrices returned by Factor$(C)$ and $S_2, S'_2$ be the matrices
  returned by Factor$(D)$.

\item Let $S =(S_1 \oplus S_2)\tilde{S}$ and $S' = \tilde{S'}(S'_1
  \oplus S'_2)$. Return the invertible scalar matrices $S$ and
  $S'$. Note that at this stage $SLS'$ has the desired atomic block
  diagonal form.
\end{enumerate}

Next we give a brief argument for proving correctness of the above
algorithm. Firstly, the algorithm declares $L$ as an atom iff $L$ is
indeed an atom. To see this, we will prove $L$ is not an atom iff
$L''$ is not an atom. Forward direction is obvious. To prove the
reverse direction of implication, let $L''$ is not an atom. Which
implies $L' = W L''$ is not an atom. $\tilde{L}$ is a linear matrix
obtained by substituting $M_i(j,k)=0$ for all $i,j,k$ in
$L'$. Clearly, $\tilde{L}$ is not an atom as $L'$ is not an
atom. Using Lemma \ref{hkv20-main-lemma} it follows that $L$ is not an
atom. So we have established $L$ is not a atom iff $L''$ is not an
atom. So if input linear matrix $L$ is an atom, the algorithm will
correctly declare it to be an atom in step 2.

Now we argue that we will get correct atomic block diagonal form in
the last step of the algorithm. Firstly, for giving recursive calls to
the Factor procedure for the matrices $C$, $D$, we must have $C, D$ to
be right monic as stated in the claim below. This is proved by the
same argument as in the proof of Theorem~\ref{cohnthm}.

%%The monicity of $C,D$ follows from monicity of $L$ as proved in the
%%Claim \ref{monic-monic} below.

\begin{claim} \label{monic-monic}
Let $L \in \fR^{d \times d}$ be a full and right monic linear matrix
such that $P'LQ' = \left(
\begin{array}{cc}
C & 0 \\
E & D
\end{array}
\right)$ where $C$ and $D$ are linear matrices of dimensions $e \times
e$, $g \times g$, respectively, such that $e+g = d$. Then both $C, D$
are right monic.
\end{claim}

By recursive calls Factor($C$) and Factor($D$) obtain matrices $S_1,
S'_1, S_2, S'_2$ such that $S_1C S'_1=C'$ and $S_2 D S'_2=D'$ are in
atomic block diagonal form.  We can write $\tilde{S}L\tilde{S'}$ as
\begin{eqnarray*}
&=& \left(
\begin{array}{cc}
C & 0 \\
E & D
\end{array}
\right)\\
&=& \left(
\begin{array}{cc}
C & 0 \\
0 & I_g
\end{array}
\right)\left(
\begin{array}{cc}
I_e & 0 \\
E & I_g
\end{array}
\right)\left(
\begin{array}{cc}
I_e & 0 \\
0 & B
\end{array}
\right)\\
&=& (S_1^{-1} \oplus I_g)(C' \oplus I_g)({S'}_1^{-1} \oplus I_g)\left(
\begin{array}{cc}
I_e & 0 \\
E & I_g
\end{array}
\right)(I_e \oplus S_2^{-1})(I_e \oplus D')(I_e \oplus {S'}_2^{-1})\\
&=&  (S_1^{-1} \oplus I_g)(C' \oplus I_g)(I_e \oplus {S'}_2^{-1})\left(
\begin{array}{cc}
I_e & 0 \\
S_2ES'_1 & I_g
\end{array}
\right)({S'}_1^{-1}\oplus I_g)(I_e \oplus D')(I_e \oplus {S'}_2^{-1})\\
&=& (S_1^{-1} \oplus I_g)(I_e \oplus S_2^{-1})(C' \oplus I_g)\left(
\begin{array}{cc}
I_e & 0 \\
S_2ES'_1 & I_g
\end{array}
\right)(I_e \oplus D')({S'}_1^{-1}\oplus I_g)(I_e \oplus {S'}_2^{-1})\\
&=& (S_1^{-1} \oplus I_g)(I_e \oplus S_2^{-1})\left(
\begin{array}{cc}
C' & 0 \\
S_2ES'_1 & D'
\end{array}
\right)({S'}_1^{-1}\oplus I_g)(I_e \oplus {S'}_2^{-1})\\
&=& (S_1^{-1} \oplus S_2^{-1})\left(
\begin{array}{cc}
C' & 0 \\
S_2ES'_1 & D'
\end{array}
\right)({S'}_1^{-1}\oplus {S'}_2^{-1}).
\end{eqnarray*}
Thus we have
\[
(S_1 \oplus S_2) \tilde{S}L\tilde{S'}(S'_1 \oplus S'_2) = \left(
\begin{array}{cc}
C' & 0 \\
S_2ES'_1 & D'
\end{array}
\right).
\]

As $C'$ and $D'$ are in atomic block diagonal form, it follows that
$\left(
\begin{array}{cc}
C' & 0 \\
S_2ES'_1 & D'
\end{array}
\right)$ is also in atomic block diagonal form. Letting $S=(S_1 \oplus
S_2) \tilde{S}$ and $S'=\tilde{S'}(S'_1 \oplus S'_2)$, it follows that
$SLS'$ is in the desired atomic block diagonal form which proves the
correctness of Factor procedure. In each call to the procedure
(excluding the recursive calls) the algorithm takes
$\poly(n,d,\log_2q)$ time. The total number of recursive calls overall
is bounded by $d$. Hence, the overall running time is
$\poly(n,d,\log_2q)$. This completes the proof of the theorem.
\end{proof}

For the factorization of $f$, we assume the stably associated full
linear matrix $L$ is left monic. After we obtain atomic block
diagonal form as in Equation \ref{eq-block-atom}, we can factorize $L$
into atomic factors by Theorem~\ref{lfact1}. Combined with
Equation~\ref{eq4} we have
\begin{equation*}
f\oplus I_s = PS'F'_1F'_2\cdots F'_r U'Q,
\end{equation*}
where each linear matrix $F'_i$ is an atom, $P$ is upper triangular
with all $1$'s diagonal, $Q$ is lower triangular with all $1$'s
diagonal, and $S'$ is scalar invertible and $U'$ is a unit. Now,
applying Lemma \ref{extract} and Theorem \ref{cnzthm} we obtain the
complete factorization of $f$ into irreducible factors. This is
summarized in the following.

\begin{theorem}\label{comzthm}  
  Let $f\in \fR$ be a polynomial given by an algebraic branching
  program as input instance of $\prob{FACT}(\F_q)$. Then there is a
  $\poly(s, \log q,|X|)$ time randomized algorithm that outputs a
  factorization $f=f_1f_2\cdots f_r$ such that each $f_i$ is
  irreducible and is output as an algebraic branching program.
\end{theorem}

Analogous to Corollary \ref{sparse-cor1}, when the polynomial is given
in a sparse representation, we have 

\begin{corollary} \label{sparse-comnz}
If $f\in\fR$ is a polynomial given as input in sparse representation
(that is, an $\F_q$-linear combination of its monomials) then in
randomized polynomial time we can compute a factorization into
irreducible factors in sparse representation.
\end{corollary}

%Let $\overline{M}=(M_1,M_2,\ldots,M_\ell)$ be a matrix substitution
%(obtained using Theorem \ref{amitsur}) such that $f(\overline{M})=W$
%is an invertible scalar matrix. Let $g$ be a polynomial defined as
%$g= W^{-1}\cdot f(Y_1+M_1, Y_2+M_2, \ldots, Y_\ell+ M_\ell) \in
%\fR^{d \times d}$ where for $i\in [\ell]$, $Y_i$ is $d \times d$
%matrix such that $(j,k)^{th}$ entry of $Y_i$ is a fresh noncommuting
%variable $Y_{i,j,k}$ for $1\leq j, k \leq d$.  Let $L \in \fR^{n
%\times n}, L'\in \fR^{nd \times nd}$ be monic linear matrices stably
%associated with $f, g$ respectively, which are obtained via Higman
%Linearization and subsequent application of Theorem \ref{}. We have
%$f \oplus I_{r} = P(L\oplus I_s)Q$, such that $r+1 = n+s$ and $P, Q$
%are upper and lower triangular invertible polynomial matrices in
%$\fR^{(n+s) \times (n+s)}$ with diagonal entries equal to $1$. So we
%have, $g \oplus I_{rd} = P'(L'\oplus I_{sd})Q'$, and $L' =
%L(\overline{M}), P' = P(\overline{M}), Q' =
%Q(\overline{M})$. Moreover $P', Q'$ are also upper and lower
%triangular invertible polynomial matrices in $\fR^{(nd+sd)\times
%(nd+sd)}$ respectively with diagonal entries equal to $1$ and
%$L'(\overline{0})= I_{nd}$.

\subsection{Factorization over small finite fields}

Finally, we briefly discuss the factorization problem over small
finite fields.  As explained in Section \ref{small-field}, the two
steps in our factoring algorithm requiring randomization can be
replaced with deterministic $\poly(s, q, |X|)$ time
computation. Furthermore, as explained in Section~\ref{small-field},
the matrix shift $(M_1, M_2, \ldots, M_n)$ required for the Theorem
\ref{lfact2} can be obtained in deterministic polynomial time such
that the entries of the matrices $M_i$ are from $\F_q$ for each $i$.
Putting it together, it gives us a deterministic factorization
algorithm for noncommutative polynomials that are input as arithmetic
formulas over $\F_q$. In summary, we have the following.

\begin{theorem}
Given as input a multivariate polynomial $f\in\F_q\angle{X}$ for a
finite field $\F_q$ by a noncommutative algebraic branching program of size
$s$, a factorization of $f$ as a product $f=f_1f_2\cdots f_r$ can be
computed in deterministic time $\poly(s,q,|X|)$, where each
$f_i\in\F_q\angle{X}$ is an irreducible polynomial that is output as
an algebraic branching program.
\end{theorem}

\section{Concluding Remarks}

In this paper we present a randomized polynomial-time algorithm for
the factorization of noncommutative polyomials \emph{over finite
  fields} that are input as \emph{algebraic branching programs}. The
irreducible factors are output as algebraic branching programs.

Several open questions arise from our work. We mention two of them.
The first question is the complexity of factorization over rationals
of noncommutative polynomials given as ABPs or arithmetic formulas. Our
approach involves the crucial use of Ronyai's algorithm for invariant
subspace comptutation which turns out to be a hard problem over
rationals.  We believe a different approach may be required for the
rational case.

The use of Higman linearization prevents us from generalizing this
approach to noncommutative polynomials given as arithmetic
circuits. We do not know any nontrivial complexity upper bound for the
factorization problem for noncommutative polynomials given as
arithmetic circuits.\\

\noindent\textbf{Acknowledgements.}~~We thank anonymous referees
and Partha Mukhopadhyay for asking about extension of our earlier
factorization algorithm for noncommutative formulas to algebraic
branching programs.

\bibliographystyle{alpha}
\bibliography{references}
\appendix

\newpage
\section{Appendix}

%\subsection{Proof of Theorem~\ref{block-higman}}

\subsection{Higman Linearization of Algebraic Branching Programs}

In this section we present a proof of Theorem~\ref{block-higman}. More
precisely, we give an efficient deterministic algorithm to compute
Higman Linearization for a non-commutative polynomial given by an
algebraic branching program. We obtain this by generalizing the Higman
Linearization process (described in Theorem~\ref{ehigman}) to what we call
Block-Higman linearization. We restate Theorem~\ref{block-higman}) and
present its proof.

%As a consequence of this theorem, it follows that
%the randomized polynomial-time factorization algorithm described in
%Section~\ref{} also works for noncommutative polynomials computed by
%algebraic branching programs.

\begin{theorem}[Block Higman Linearization]\label{blockhigthm}
  Given as input an algebraic branching program of size $s$ computing
  a noncommutative polynomial $f\in\fR$, we can compute in
  deterministic polynomial time matrices
  $P, Q\in\fR^{\ell \times \ell}$ and a linear matrix
  $L\in\fR^{\ell\times \ell}$ such that
\begin{equation}\label{higeq}
%PAQ ~=~  L \text{ where } A ~=~ 
P\left(
\begin{array}{c|c}
f & 0 \\
\hline
0 & I_{\ell-1}
\end{array}
\right)Q ~=~L
\end{equation}
with $P$ upper triangular, $Q$ lower triangular, and the diagonal
entries of both $P$ and $Q$ are all $1$'s (hence, $P$ and $Q$ are both
units in $\fR^{\ell \times \ell}$). Futhermore, the algorithm computes
the entries of matrices $P$ and $Q$ as algebraic branching programs.
The entries of $P^{-1}$ and $Q^{-1}$ are also computable as algebraic
branching programs.

%of size $O(\ell)$ where
%$\ell$ is $O(s)$.
\end{theorem}

\begin{proof}
  Let $f\in \fR$ be the input noncommutative polynomial of degree $d$
  computed by an ABP of size $s$ and $d+1$ layers, where the $i^{th}$
  layer has say $n_i$ nodes. Then we have linear matrices
  $A_1,A_2,\ldots,A_d$ such that
  \[
    f=A_1\cdot A_2\cdots A_d,
  \]
  where $A_i$ is $n_i\times n_{i+1}$ for each $i$ and $n_1=n_{d+1}=1$.
  Note that the $(j,k)^{th}$ entry of $A_i$ is the linear form
  labeling the edge from $j^{th}$ node in layer $i$ to $k^{th}$
  node in layer $i+1$.

  We will prove the following more general result, dropping the
  constraint that $n_1=n_{d+1}=1$. Suppose $A_1,A_2,\ldots,A_d$ are
  linear matrices of compatible dimensions ($A_i$ is
  $n_i\times n_{i+1}$ for each $i$) such that the matrix product
\[
  M=A_1\cdot A_2\cdots A_d
\]  
is well defined. The algorithm we will describe will compute matrices
$P, Q, L$ such that

\begin{equation}\label{higeq}
%PAQ ~=~  L \text{ where } A ~=~ 
P\left(
\begin{array}{c|c}
M & 0 \\
\hline
0 & I_{\ell-1}
\end{array}
\right)Q ~=~L
\end{equation}

and $P, Q, L$ have the properties as stated in the theorem.

%\end{proof}  

%First we describe a single step of Block-Higman linearization when the
%input matrix $M$ is product of two linear matrices.

The proof is by an easy induction on number $d$ of linear matrices
whose product is $M$. We set up this induction by describing a single
step of Block-Higman linearization writing the input matrix $M=AB$,
where $A=A_1\cdot A_2\cdots A_{d-1}$ and $B=A_d$ is a linear matrix.

So, let $M$ be the $r \times t$ matrix where $M=A\cdot B$ with
$A \in \fR^{r\times s}$ and $B \in \fR^{s\times t}$, and $B$ is a linear
matrix. We apply the following steps to $M$.

\begin{itemize}
\item Expand $M$ to a $(r+s) \times (t+s)$ matrix of the following
  shape by adding $s$ new rows and $s$ new columns, with the bottom
  right diagonal block being $I_s$ and the remaining entries zero to
  obtain the following:
\[
\left[
\begin{array}{c|c}
M & 0 \\
\hline
0 & I_s
\end{array}
\right]. 
\]
\item Use suitable block row and column operations to transform the
  matrix as follows
\[
\left(
\begin{array}{cc}
AB & 0 \\
0 & I_s
\end{array}
\right)\rightarrow
\left(
\begin{array}{cc}
AB & A \\
0 & I_s
\end{array}
\right)\rightarrow
\left(
\begin{array}{cc}
0 & A \\
-B & I_s
\end{array}
\right)
\]

Here the first step is realized by computing the matrix product
$A\cdot [0~|~I_s]$ and adding it to the respective blocks of the first
$r$ rows. The second step is realized by computing the matrix product
\[
 \left(
\begin{array}{c}
  A \\
  I_s
\end{array}
\right)\cdot (-B)
\]
and adding this to the respective blocks of the first $t$ columns.
These two steps are realized by left multiplication by $\left(
\begin{array}{c|c}
I_r & A \\
\hline
0 & I_s
\end{array}
\right)$ and right multiplication by $\left(
\begin{array}{c|c}
I_t & 0 \\
\hline
 -B & I_s
\end{array}
\right).$
\end{itemize}
So we have
\[\left(
\begin{array}{cc}
I_r & A \\
%\hline
0 & I_s
\end{array}
\right)
\left(
\begin{array}{cc}
AB & 0 \\
0 & I_s
\end{array}
\right)
\left(
\begin{array}{cc}
I_t & 0 \\
%\hline
 -B & I_s
\end{array}
\right) = 
\left(
\begin{array}{cc}
0 & A \\
-B & I_s
\end{array}
\right).
\]
In the above, we note that the row operation matrix $\left(
\begin{array}{cc}
I_r & A \\
%\hline
0 & I_s
\end{array}
\right)$ is upper triangular with all diagonal entries $1$ where as the column operation matrix $\left(
\begin{array}{cc}
I_t & 0 \\
%\hline
 -B & I_s
\end{array}
\right)$ is lower triangular with all diagonal entries $1$. 

Crucially, if $A=A_1\cdot A_2\cdots A_{d-1}$ and $B=A_d$, the
resulting matrix has only $A$ as the nonlinear block which is a
product of $d-1$ linear matrices. We can now apply induction to
the matrix $A$ to obtain the Block-Higman linearization of $M$
as claimed. 

We describe the intermediate steps of the induction in more detail, in
order to see the final shape of the linear matrix.

Let $M= A_1A_2\ldots A_d$ where $A_i \in \fR ^{n_i \times n_{i+1}}$,
are linear matrices. At the $i^{th}$ step of the induction
Block-Higman linearization transforms a matrix of the form $\left(
\begin{array}{cc}
* & A_1A_2\ldots A_{d-i} \\
%\hline
* & *
\end{array} 
\right)$ into a matrix of the form $\left(
\begin{array}{cc}
* & A_1A_2\ldots A_{d-i-1} \\
%\hline
* & *
\end{array} 
\right)$ 
where each $*$ indicates matrix blocks with linear entries. After
$d-1$ steps of Block-Higman linearization we obtain a linear matrix
which is an associate of $M$. 

In more detail, let $M_0=A_1A_2\ldots A_d, P_0=I_{n_1}$,
$Q_0=I_{n_{d+1}}$ and $t_0=0$. We have
$P_0 (M \oplus I_{t_0}) Q_0 = M_0$. Inductively, assume that we have
upper triangular matrix $P_i$ with all diagonal entries $1$ and a
lower triangular matrix $Q_i$ with all diagonal entries $1$ such that
\[ P_i \left(
\begin{array}{cc}
M & 0 \\
%\hline
0 & I_{t_i}
\end{array} 
\right)Q_i = M_i.
\] 
Here, $M_i$ is a polynomial matrix which has top right block equal to
$A_1A_2 \ldots A_{d-i}$ and the other entries of $M$ are linear and
the entries of $P_i$ and $Q_i$ are all computable by ABPs. Let
$t_{i+1}=t_i + n_{d-i}$. Clearly,
$(P_i \oplus I_{n_{d-i}}) (M \oplus I_{t_{i+1}})(Q_i \oplus
I_{n_{d-i}})= M_i \oplus I_{n_{d-i}}$. Let $M_i'$ be a matrix obtained
from $M_i$ by replacing top right block by $0$. Let $M_{i+1}$ be a
$2\times 2$ block matrix with $M_i'$ as top left block and the
structure as shown below
\[
M_{i+1}= \left(
\begin{array}{ccccc|c|c}
 &&&&& 0& A_1A_2\ldots A_{d-i-1} \\
\hline
%&&&&& & 0\\
&&&&& & 0\\
&&&&& & \vdots\\
&&&&& & 0\\
%&&&&& & 0\\
\hline
0& 0& \ldots & 0 & 0& -A_{d-i} & I_{n_{d-i}}
\end{array} 
\right),
\]
where matrix blocks $-A_{d-i}$ and $A_1A_2\ldots A_{d-i-1}$ align with top right $0$ block in $M_i'$. Now we define suitable block row and column operations which transforms matrix $M_i \oplus I_{n_{d-i}}$ to $M_{i+1}$. 

By applying the Block-Higman linearization step we will obtain

% Let $P'$ be a block row operation which adds
% $A_1A_2 \ldots A_{d-i-1}$ times
% the% bottom row block consisting of $n_{d-i}$ rows to the top row 
% block consisting of first $n_1$ rows. That is $P'$ be
% $(n_1+T_i+ n_{d-i}) \times (n_1+T_i+ n_{d-i})$ upper triangular
% matrix with diagonal entries $1$ and top right block of
% $P'$ (with first $n_1$ rows and last $n_{d-i}$ columns)
% equal to $A_1A_2 \ldots A_{d-i-1}$.

%Similarly let $Q'$ be a block column operation which subtracts last
%column block% (consisting of last $n_{d-i}$ columns) times $A_{d-i}$
%from the column block consisting of $n_{d-i+1}$ columns prior to the
%last column block consisting of $n_{d-i}$ columns. That is $Q'$ is a
%$(n_{d+1} + T_i + n_{d-i})\times (n_{d+1} + T_i + n_{d-i})$ lower
%triangular matrix with diagonal entries $1$ and block situated in last
%$n_{d-i}$ rows and $n_{d-i+1}$ columns just before the last $n_{d-i}$
%columns equals to $-A_{d-i}$.  

\[
P'(M_i \oplus I_{n_{d-i}})Q' = M_{i+1},
\] 
where $P'$ and $Q'$ are upper and lower triangular matrices performing
the block row and column operations and their entries are computable
by ABPs. Letting $P_{i+1} = P' (P_i \oplus I_{n_{d-i}})$ and
$Q_{i+1}= (Q_i \oplus I_{n_{d-i}}) Q'$ we get
\[
P_{i+1} (M \oplus I_{t_{i+1}}) Q_{i+1} = M_{i+1}
\]
where $P_{i+1}$ and $Q_{i+1}$ are upper and lower triangular matrices
with diagonal entries $1$, the top right block (consisting of top
$n_1$ rows and last $n_{d-i}$ columns) of $M_{i+1}$ is
$A_1A_2 \ldots A_{d-i-1}$ and all other entries of $M_{i+1}$ are
linear. Continuing thus, we obtain upper and lower triangular matrices
$P_{d-1}$ and $Q_{d-1}$ with all diagonal entries $1$ such that
$P_{d-1} (M \oplus I_{t_{d-1}})Q_{d-1}= M_{d-1}$ which is a linear
matrix. Moreover, it is easy to see that entries of $P_{d-1}$ and
$Q_{d-1}$ are given by polynomial size ABPs (of $O(s^2)$ size to be
precise).

Carefully observing the shapes of the matrices $M_i$ we note that the
final linearized matrix $M_{d-1}$ has the form
\[
\left(
\begin{array}{ccccccc}
0 & 0 & 0 & 0 & \ldots & 0 & A_1\\
%\hline
A_d & I_{n_{d}} & 0 & 0 & \ldots & 0 & 0\\
0 & A_{d-1} & I_{n_{d-1}} & 0 & \ldots & 0 & 0 \\
0 & 0 & A_{d-2} & I_{n_{d-2}}& \ldots & 0 & 0 \\
0 & 0 & \ldots & \vdots & \ddots & \vdots & 0\\
0& 0 & 0& 0 & \ldots & A_2 & I_{n_2}
\end{array} 
\right)
\]

$M_{d-1}$ is a $d\times d$ block matrix with
\begin{itemize}
\item[-]  $n_1\times n_2$ sized top right block in $A_1$.
\item[-]  $(i,i)^{th}$ block is $I_{n_{d+2-i}}$ for $2\leq i \leq d$.
\item[-] $(i,i-1)^{th}$ block is $A_{d+2-i}$ for $2\leq i \leq d$.
\item[-] all other entries are $0$.
\end{itemize}

\end{proof}

\subsection{Missing proofs from Section~\ref{prelim}}

\begin{proofof}{Theorem~\ref{cohnthm}}
Let $C\in\fR^{d\times d}$ be a full and right monic linear
matrix. Suppose Equation~\ref{eq-cohnthm} holds for some invertible
scalar matrices $S,S'$. Then we can write

\[
SCS'=  \left(
\begin{array}{cc}
A & 0 \\
D & B
\end{array}
\right) =   \left(
\begin{array}{cc}
A & 0 \\
0 & I
\end{array}
\right)\cdot 
  \left(
\begin{array}{cc}
I & 0 \\
D & I
\end{array}
\right)\cdot
\left(
\begin{array}{cc}
I & 0 \\
0 & B
\end{array}
\right).
\]
Since $C$ is right monic and $S, S'$ are invertible scalar matrices
the linear matrix $SCS'=\left(\begin{array}{cc} A & 0 \\ D & B
\end{array}\right)$ is also full and right monic. Writing it as
\[
\left(
\begin{array}{cc}
A & 0 \\
D & B
\end{array}
\right) =   \left(
\begin{array}{cc}
A_0 & 0 \\
D_0 & B_0
\end{array}
\right) + \sum_{i=1}^n \left(
\begin{array}{cc}
A_i & 0 \\
D_i & B_i
\end{array}
\right)\cdot x_i,
\]
it means the matrix
\[
\left[
\begin{array}{cc|}
A_1 & 0 \\
D_1 & B_1
\end{array}
\begin{array}{cc|}
A_2 & 0 \\
D_2 & B_2
\end{array}
 \ldots
\begin{array}{|cc}
A_n & 0 \\
D_n & B_n
\end{array}
\right]
\]
is full row rank. With suitable row operations applied to the above we
can see that both $[A_1 A_2\ldots A_n]$ and $[B_1 B_2 \ldots B_n]$ are
full row rank. Therefore, both $A$ and $B$ are full right monic
matrices hence they are nonunits by Lemma~\ref{monic-nonunit}.  Hence $A\oplus I$ and $B\oplus I$ are both non-units  which
implies that the factorization of $SCS'$ is nontrivial and hence $C$
is not an atom. 

Conversely, suppose $C$ is not an atom and $C=F\cdot G$ is a
nontrivial factorization. That means both $F$ and $G$ are full and
non-units. As $C$ is a linear matrix, applying \cite[Lemma
  5.8.7]{Cohnfir} we can assume that both $F$ and $G$ are linear
matrices. Now, since $F$ is a full linear matrix, by
Theorem~\ref{full-monic} (and Remark~\ref{remark-full-monic}) there
are a scalar invertible matrix $S_1$ and polynomial matrix $U_1$,
which is a unit, such that $S_1FU_1 = A\oplus I$ such that $A$ is
\emph{left} monic. Therefore, we have
\[
S_1 C = S_1 F U U^{-1} G = \left(
\begin{array}{cc}
A & 0 \\
0 & I
\end{array}
\right) \cdot \left(
\begin{array}{cc}
G'_1 & G'_3 \\
G'_2 & G'_4
\end{array}
\right) = \left(
\begin{array}{cc}
AG'_1 & AG'_3 \\
G'_2 & G'_4
\end{array}
\right).
\]
As $S_1 C$ is a linear matrix and $A$ is a left monic linear matrix we
can assume that $G'_1$ and $G'_3$ are scalar matrices. Since $S_1 C$ is
full rank, it forces the matrix $[G'_1 G'_3]$ to be full row rank (say
$r$, where $A$ is $r\times r$). Therefore, there is an invertible
scalar matrix $S'$ such that $[G'_1 G'_3]S' = [I_r 0]$. Putting it together,
we get the factorization
\[
SCS' = \left(
\begin{array}{cc}
A & 0 \\
0 & I
\end{array}
\right) \cdot
\left(
\begin{array}{cc}
I_r & 0 \\
G''_2 & G''_4
\end{array}
\right) =  \left(
\begin{array}{cc}
A & 0 \\
G''_2 & G''_4
\end{array}
\right)
\]
as claimed by the theorem.
\end{proofof}

\subsection{Proof of Lemma~\ref{hkv20-main-lemma}}

We present a self-contained proof of Lemma~\ref{hkv20-main-lemma} of
\cite{HKV20}.

\begin{definition}\label{def-setminus}
Let $U, V \subseteq \F^D$ be subspaces of $\F^D$ and $d= \dim U$. Fix
a basis $u_1, u_2, \ldots, u_\ell \in \F^D$ for $U \cap V$ and extend
it to a basis $u_1, u_2, \ldots, u_\ell, u_{\ell+1}, \ldots, u_d$ for
$U$. Further, let $u_1, u_2, ..., u_D$ be a basis for $\F^D$ obtained
by extending the above basis for $U$. Then $U\setminus V$ is defined
as $span(u_{\ell+1}, u_{\ell+2}, \ldots, u_d)$, i.e.
\[
U\setminus V = \{ \sum_{i= \ell+1}^{d} \alpha_i u_i | \alpha_i \in \mathbb{F} \textrm{ for } \ell < i \leq d \} 
\]
\end{definition}

Clearly $\dim U\setminus V = \dim U - \dim U \cap V$. Notice that
although the subspace $U\setminus V$ is basis dependent, the number
$\dim U\setminus V$ is independent of the construction of $U\setminus
V$.

\begin{definition}\label{projection}
Let $\mathcal{U}=\{U_1, U_2, \ldots, U_d \}$ be a collection of
subspaces of $\mathbb{F}^{D}$. For each $i \in [d]$ define
$\hat{U_i}^{(\mathcal{U})} = U_i \setminus (\sum_{k \neq i} U_k)$ as
above with respect to fixed bases for the subspaces.
\end{definition}

We first prove a technical lemma, essentially using the
inclusion-exclusion principle.

\begin{lemma}\label{hkv-rank-lemma}
Let $\mathcal{U}= \{U_1, U_2, \ldots, U_d\}$ be a collection of
subspaces of $\mathbb{F}^D$ for $d\geq 1$. Then
\[
\sum_{i=1}^{d}~ \left[\dim U_i + \dim \hat{U}_i^{(\mathcal{U})}
  \right]~~ \geq ~~2~\cdot~\dim \sum_{i=1}^d U_i.
\]
\end{lemma}

\begin{proof}
The proof will be by induction on $d$. The base case, $d=1$, is
obvious. Suppose it is is true for all $t <d$. I.e. for any subspace
collection $\mathcal{V} = \{V_1, V_2, \ldots, V_t\}$ we have
\[
\sum_{i=1}^{t}~ \left[\dim V_i + \dim  \hat{V}_i^{(\mathcal{V})} \right]~~ \geq ~~2~\cdot ~\dim\sum_{i=1}^t V_i.
\]
Letting $V_i= U_i$ for $1\leq i \leq d-2$ and $V_{d-1}= U_{d-1}+U_d$
in the above, we have
\[
\sum_{i=1}^{d-1}~ \left[\dim V_i + \dim \hat{V}_i^{(\mathcal{V})}
  \right] ~~\geq ~~2~\cdot~ \dim\sum_{i=1}^{d-1} V_i = 2\cdot\dim\sum_{i=1}^d U_i. 
\]

For the induction we need to show that $\sum_{i=1}^{d-1} (\dim V_i +
\dim \hat{V}_i^{(\mathcal{V})})\leq \sum_{i=1}^{d} (\dim U_i + \dim
\hat{U}_i^{(\mathcal{U})})$.

Now, \[\sum_{i=1}^{d-1}~ (\dim V_i + \dim
  \hat{V}_i^{(\mathcal{V})})\] is
\begin{eqnarray*}
 &=& \dim V_{d-1} + \dim \hat{V}_{d-1}^{(\mathcal{V})} + \sum_{i=1}^{d-2}~ \left[  ~  \dim U_i + \dim(U_i \setminus ( U_{d-1}+U_{d} + \sum_{k \neq i, k<d-1}U_k ))~ \right]  \\
 &=&  \dim V_{d-1} + \dim \hat{V}_{d-1}^{(\mathcal{V})} + \sum_{i=1}^{d-2}~ \left[ ~\dim U_i + \dim( U_i \setminus \sum_{k \neq i, k\leq d} U_k) \right] \\
 &=&  \dim V_{d-1} + \dim \hat{V}_{d-1}^{(\mathcal{V})}+ \sum_{i=1}^{d-2} ~\left[ \dim U_i + \dim \hat{U}_i^{(\mathcal{U})} \right]\\
 &=& \dim U_{d-1} + \dim U_d -dim( U_{d-1} \cap U_d ) + \dim \hat{V}_{d-1}^{(\mathcal{V})}+ \sum_{i=1}^{d-2}~\left [ \dim U_i + \dim \hat{U}_i^{(\mathcal{U})} \right] \\
  &=& \left[ \sum_{i=1}^{d} \dim U_i \right] + \left[ \sum_{i=1}^{d-2}~ \dim \hat{U}_i^{(\mathcal{U})} \right] + \dim \hat{V}_{d-1}^{(\mathcal{V})}- \dim (U_{d-1} \cap U_d).
\end{eqnarray*}

%%Hence, we will get \[\sum_{i=1}^{d-1}~\left[\dim V_i + \dim
%%\hat{V}_i^{(\mathcal{V})} \right]~~ \leq~~ \%sum_{i=1}^{d}~ \left[
%%\dim U_i + \dim \hat{U}_i^{(\mathcal{U})} \right]\]

Hence, to complete the proof it suffices to show the following claim.

\begin{claim}
\[ \dim \hat{V}_{d-1}^{(\mathcal{V})}~ \leq~ \dim  \hat{U}_{d-1}^{(\mathcal{U})}  + \dim  \hat{U}_d^{(\mathcal{U})} + \dim (U_{d-1} \cap U_d) \]
\end{claim}
\begin{proofof}{Claim}
Let $T= U_1 + U_2 + \ldots + U_{d-2}$. Let $D, D_1, D_2$, and $D_3$ denote
dimensions of $T+U_{d-1}+U_d, U_{d-1}, U_d$, and $T$ respectively.

We have 
\begin{eqnarray*}
\dim \hat{V}_{d-1}^{(\mathcal{V})} &=& \dim (U_{d-1}+U_d)-\dim ((U_{d-1}+U_d) \cap T)\\
&=&\dim (U_{d-1}+U_d)-\dim (U_{d-1}+U_d)-\dim (T) + \dim (U_{d-1}+U_d+T)\\
&=& D-D_3
\end{eqnarray*}

%Hence,
%\begin{equation}\label{eqnV}
%\dim \hat{V}_{d-1}^{(\mathcal{V})}=D-D_3
%\end{equation}

Similarly,
\begin{eqnarray*}
\dim  \hat{U}_{d-1}^{(\mathcal{U})}  &=& \dim (U_{d-1}) - \dim (U_{d-1} \cap (T+ U_d))\\
&=&D_1 + \dim (T+U_{d-1}+U_d)- \dim (U_{d-1}) - \dim (T+U_d)\\
&=&D_1 +D -D_1- \dim (T)-\dim (U_d)+\dim (T \cap U_d)\\
&=& D-D_3-D_2 + \dim(T \cap U_d). 
\end{eqnarray*}
Likewise, we also have
\[
\dim ( \hat{U}_{d}^{(\mathcal{U})} )=D-D_3-D_1+ \dim (T \cap U_{d-1}).
\]
It is clear that the claim is equivalent to
\[
 D \geq D_1 + D_2 + D_3 - \dim (T \cap U_d) - \dim (T \cap U_{d-1})-
 \dim (U_{d-1} \cap U_d)
 \]
which follows immediately from the Inclusion-Exclusion Principle. 
\end{proofof}
\end{proof}

%The proof \cite{HKV20}, Lemma \ref{hkv-rank-lemma} is inferred from a
%combinatorial statement but the argument is unclear. For this reason,
%we have given above a self contained inductive proof.

Now we will present a complete proof for Lemma \ref{hkv20-main-lemma}
from \cite{HKV20} where the proof is sketchy.\smallskip
 
{\bf Proof of Lemma \ref{hkv20-main-lemma} }

Let $L=A_0 + \sum_{i=1}^n A_i x_i$ Where $A_i \in \mathbb{F}_q^{d
  \times d}$ for $i \in [n]$. So, $L'= A_0 \otimes I_\ell +
\sum_{i=1}^{n} A_i \otimes Y_i$ . From standard properties of the
Kronecker product of matrices, there are a row permutation matrix $R$
and a column permutation matrix $C$ such that $L''= RL'C = I_\ell
\otimes A_0 + \sum_{i=1}^n Y_i \otimes A_i$. So $L'' = I_\ell \otimes
A_0 + \sum_{i=1}^n \sum_{1 \leq j, k \leq \ell} (E_{j,k} \otimes A_i)~
y_{i,j,k}$, where $E_{j,k}$ is a a matrix with $(j,k)^{th}$ entry one
and rest all entries equal to zero.
     
 We have $G L' H = \left(
\begin{array}{c|c}
A' & 0 \\
\hline
D' & B'
\end{array}
\right)$
Where $A'$ is $d' \times d'$ linear matrix for $d'>0$, $B'$ is $d''
\times d''$ linear matrix with $d' + d'' = d \ell$. Hence, $GR^{-1}
L''C^{-1} H = \left(
\begin{array}{c|c}
A' & 0 \\
\hline
D' & B'
\end{array}
\right)$.  Let $GR^{-1}=P_0$ and $C^{-1}H=Q_0$. Let $[ P_1 P_2 \ldots
  P_\ell]$ be the (full row rank) matrix obtained by picking the top
$d'$ rows of $P_0$ where each $P_i$ is $d' \times d$ scalar
matrix. Similarly let $[ Q_1^T Q_2^T \ldots Q_\ell^T]^T$ be the (full column
rank) matrix obtained by picking the rightmost $d''$ columns of $Q_0$
where each $Q_i$ is $d \times d''$ scalar matrix. Clearly,
\[   [ P_1 P_2 \ldots P_\ell] L'' [ Q_1^T Q_2^T \ldots Q_\ell^T]^T = 0, \]
which implies,
\[   [ P_1 P_2 \ldots P_\ell]\left[ I_\ell \otimes A_0 + \sum_{i=1}^n \sum_{1 \leq j, k \leq \ell} (E_{j,k} \otimes A_i)~ y_{i,j,k} \right] [ Q_1^T Q_2^T \ldots Q_\ell^T]^T = 0. \]

Equating the coefficients of each $y_{i,j,k}$ to zero we get the following.

\begin{equation} \label{eqn-A0}
\sum _{i=1}^\ell P_i A_0 Q_i = 0.
\end{equation} 

\begin{equation} \label{eqn-An0}
P_j A_i Q_k =0 \textrm{ for each } i>0 \textrm { and } 1\leq j,k \leq
\ell.
\end{equation} 

For each $i\in [\ell]$ the matrix $P_i$ is a linear transformation
from $\F^d$ to $\F^{d'}$. Let $U_i= \range(P_i)=\{P_i u | u \in
\F^{d}\}$ for each $i$, and $\mathcal{U}= \{U_1, U_2, \ldots, U_{\ell}
\}$. Let $T_i = \sum _{j \neq i} U_j$ for $i\in [\ell]$. Clearly,
$U_1+ U_2 + \ldots + U_\ell= \mathbb{F}_q^{d'}$ as $[ P_1 P_2\ldots
  P_\ell]$ is a full row rank matrix. For $i\in[\ell]$, let
$\hat{P}_i$ be a linear transformation from $\mathbb{F}^{d'}$ to
$\mathbb{F}^{d'}$ defined as follows. Fix a basis $u_{i,1}, u_{i,2},
\ldots, u_{i,r_i}$ of the subspace $U_i \cap T_i$. Extend it to a
basis $u_{i,1}, u_{i,2}, \ldots, u_{i,r_i}, u_{i,r_i+1}, \ldots,
u_{i,k_i}$, $k_i \geq r$, for $U_i$. Further, extend this basis of
$U_i$ to a complete basis $u_{i,1}, u_{i,2}, \ldots, u_{i,d'}$, for
$\mathbb{F}^{d'}$, where $d' \geq k_i$. For any vector $u =
\sum_{j=1}^{d'} \alpha_j u_{i,j}$ in $\mathbb{F}^{d'}$ let $\hat{P}_i
(u) = \sum_{j=r+1}^k \alpha_j u_{i,j}$. So, $\hat{P}_i(u)$ is the
vector obtained by projecting to the subspace $U_i\setminus T_i$
(which is defined w.r.t.\ the above basis). Hence,
$\hat{P}_i(u_{i,t})= u_{i,t}$ for $r_i < t \leq k_i$ and
$\hat{P}_i(u_{i,t})= 0$ otherwise. This defines a $d' \times d'$
matrix for each $\hat{P}_i$ for $i \in [\ell]$, which we also refer to
as $\hat{P}_i$ by abuse of notation. From the Definition
\ref{projection}, it follows that $\range(\hat{P}_i) =
\hat{U}_i^{(\mathcal{U})}$, so $\rank(P_i)=\dim
\hat{U}_i^{(\mathcal{U})}$.  Clearly, $\rank(\hat{P}_i P_i) =
\rank(\hat{P}_i)$ for $i \in [\ell]$. Now, by Lemma
\ref{hkv-rank-lemma} applied to the collection $\mathcal{U}= \{U_1,
U_2, \ldots, U_{\ell} \}$ we get
\[
\sum_{i=1}^{\ell}~ \left[ \rank P_i ~+~ \rank \hat{P}_i \right]~~ \geq
~~ 2\cdot ~ \dim \sum_{i=1}^{\ell} \range(P_i) = 2d'.
\]      

Similarly, each $Q_i:\F^{d''}\to \F^d$ is a linear map. We can define
the corresponding linear maps $\hat{Q}_i : \mathbb{F}^{d''} \to
\mathbb{F}^{d''}$ and associated $d'' \times d''$ sized matrices and
we will have $\rank Q_i\hat{Q}_i = \rank \hat{Q_i}$ for each
$i$. Applying the above argument we will get
 \[
\sum_{i=1}^{\ell}~ \left[ \rank Q_i ~+~ \rank \hat{Q}_i \right]~~ \geq
~~ 2\cdot ~ \dim \sum_{i=1}^{\ell} \range(Q_i) = 2d''.
\] 

Adding the two inequalities yields

\begin{equation}\label{eqn-php}
\sum_{i=1}^{\ell}~ \left( \rank \hat{P}_i ~+~ \rank Q_i \right)~ +  \left( \rank {P}_i ~+~ \rank \hat{Q}_i \right) ~~\geq ~~ 2\cdot ~ (d'+d'')= 2d\ell.
\end{equation}
From the Equation \ref{eqn-php} we would like to prove the following Claim.

\begin{claim} \label{php-simplify}
There exist index $i \in[\ell]$ such that $\rank \hat{P}_i ~+~ \rank Q_i \geq d$ and $\rank \hat{P}_i$, $\rank Q_i>0$ or $\rank {P}_i ~+~ \rank \hat{Q}_i \geq d$ and $\rank {P}_i$, $\rank \hat{Q}_i>0$. 
\end{claim}
First we complete the proof of the Lemma \ref{hkv20-main-lemma} assuming the Claim \ref{php-simplify}.
Without loss of generality, let index $i=1$ satisfies the Claim \ref{php-simplify} and further, let $\rank
\hat{P}_1 + \rank Q_1 \geq d$ with $\rank \hat{P}_1, \rank Q_1 >0$
(other case handled similarly).

Equation \ref{eqn-A0} implies that $\range(P_1 A_0 Q_1) \subseteq
T_1$, also clearly, $\range(P_1 A_0 Q_1) \subseteq \range(P_1)=
U_1$. Which implies $\range(P_1 A_0 Q_1)\subseteq U_1 \cap T_1$. Hence
$\hat{P}_1 P_1 A_0 Q_1 =0$. Equation \ref{eqn-An0} implies that,
$\hat{P}_1 P_j A_i Q_k= 0$ for all $i \geq 1$ and $1\leq j, k \leq
\ell$. So we get $\hat{P}_1 P_1 L Q_1 = 0$. Now $\rank(\hat{P}_1P_1) =
\rank(\hat{P}_1)\geq 1$, $\rank(Q_1) \geq 1$ and $\rank(\hat{P}_1
P_1)+\rank(Q_1) \geq d$. It follows that there exist $0 < e' \leq
\rank(\hat{P}_1 P_1)$ and $0 < e'' \leq \rank(Q_1)$ with $e'+e'' =
d$. By choosing $e'$ linearly independent rows of $\hat{P}_1P_1$ and
$e''$ linearly independent columns of $Q_1$ we obtain full row rank
matrix $U' \in \mathbb{F}_q^{e' \times d}$ and a full column rank
matrix $V' \in \mathbb{F}_q^{d \times e''}$ respectively.  Now we
extend $U'$ to a $d \times d$ matrix $U$ by adding any $d-e'$ linearly
independent rows such that $U$ is invertible. Similarly we extend $V'$
to a $d \times d$ matrix $V$ by adding any $d-e''$ linearly
independent columns such that $V$ is invertible. We clearly have $ULV
=\left(
\begin{array}{c|c}
A & 0 \\
\hline
D & B
\end{array}
\right)$ for some linear matrices $A, D, B$ such that $A \in \fR^{e' \times e'}$ and $B \in \fR^{e'' \times e''}$ with $0< e', e''$ and $e' + e'' =d$ as required. This completes the proof of Lemma \ref{hkv20-main-lemma}.
\\

\begin{proofof}{Claim~\ref{php-simplify}}

Each $P_i$ has $d$ columns and each $Q_i$ has $d$ rows. Thus, $\rank
P_i\le d$ and $\rank Q_i\le d$. Also, $\rank \hat{P}_i \leq \rank P_i$
and $\rank \hat{Q}_i \leq \rank Q_i$. Hence, $\rank \hat{P}_i + \rank
Q_i\le 2d$ and $\rank P_i + \rank \hat{Q}_i\le 2d$ for each $i$.

It follows from Inequality \ref{eqn-php} that if there is an $i$ for
which either $\rank \hat{P}_i + \rank Q_i <d$ or $\rank {P}_i + \rank
\hat{Q}_i< d$ then there must be an index $j$ such that either $\rank
\hat{P}_j + \rank Q_j > d$ or $\rank {P}_j + \rank \hat{Q}_j > d$. Two
cases arise:

\begin{enumerate}
\item for all $j\in [\ell]$, $\rank \hat{P}_j + \rank Q_j = d$ and
  $\rank {P}_j + \rank \hat{Q}_j = d$.
\item there is $j \in [\ell]$ with either $\rank \hat{P}_j + \rank Q_j
  > d$ or $\rank {P}_j + \rank \hat{Q}_j > d$.
\end{enumerate}

Suppose the first case occurs. It has the following two subcases.

\begin{enumerate}
\item[(a)] for all $j \in[\ell]$, $\rank \hat{P}_j=0$ or $\rank Q_j = 0$ and $\rank {P}_j=0$ or $\rank \hat{Q}_j=0$.
\item[(b)] there is $j\in [\ell]$ such that $\rank \hat{P}_j, \rank
  Q_j >0$ or $\rank {P}_j, \rank \hat{Q}_j> 0$
\end{enumerate}

First, consider Case 1(a). Note that $\rank \hat{P}_j=0$ implies
$\rank Q_j =d$. And $\rank Q_j=0$ implies $\rank \hat{P}_j=d$, which
implies $\rank P_j =d$. Thus, either $\rank P_j =d$ or $\rank Q_j =d$
for every $j$. Moreover, Case 1(a) also implies $\rank P_j, \rank
Q_j\in \{0,d\}$ for each $j$.

Now as $[P_1 P_2 \ldots P_\ell]$ has full row rank and $[Q_1^T
  Q_2^T~\ldots Q_\ell^T]^T$ has full column rank, Case 1(a) implies
that there are indices $j,k \in [\ell]$ such that $P_j$ and $Q_k$ are
both rank $d$ matrices. As $P_j$ is full column rank matrix, there is
a $d \times d'$ matrix $P_j'$ such that $P_j' P_j = I_d$. Similarly,
there is a $d'' \times d$ matrix $Q_k'$ such that $Q_k Q_k'=I_d$. Now
from Equation \ref{eqn-An0} we know that $P_j A_i Q_k = 0$ for all
$i$, $1\leq i \leq n$. Hence $P_j' P_j A_i Q_k Q_k' = 0$ for all $i$,
$1\leq i \leq n$. Consequently, $A_i=0$ for $1\le i\le n$ which is a
contradiction to the lemma statement. Hence case 1(a) cannot occur.

If case 1(b) or 2 holds then for some index $j\in [\ell]$ either
$\rank \hat{P}_j + \rank Q_j \geq d$ with $\rank \hat{P}_j$, $\rank
Q_j >0$ or $\rank {P}_j + \rank \hat{Q}_j \geq d$ with $\rank {P}_j$,
$\rank \hat{Q}_j>0$. 
\end{proofof}

\end{document}